\documentclass[reprint,aps,prb]{revtex4-2} 
\usepackage{amsmath,amssymb,bm,bbm,mathtools}
\usepackage{amsthm}
\usepackage{graphicx}
\usepackage[colorlinks=true, allcolors=blue]{hyperref}
\usepackage{booktabs}
\usepackage{array}
\usepackage{physics}
\usepackage{dcolumn}
\newtheorem{theorem}{Theorem}
\newtheorem{proposition}{Proposition}
\newcommand{\R}{\mathcal{R}}
\newcommand{\Lx}{L_x}
\newcommand{\Ly}{L_y}
\newcommand{\aR}{\alpha_R}
\newcommand{\aI}{\alpha_I}
\newcommand{\epnu}{\epsilon_\nu}

\begin{document}
\title{Non-Hermitian skin effect induced by spatial noncommutativity}
\author{Zheng Wei}
\affiliation{Center for Advanced Quantum Studies, School of Physics and Astronomy, Beijing Normal University, Beijing 100875, China}
\affiliation{Key Laboratory of Multiscale Spin Physics (Ministry of Education), Beijing Normal University, Beijing 100875, China}

\author{Su-Peng Kou}
\email{spkou@bnu.edu.cn}
\affiliation{Center for Advanced Quantum Studies, School of Physics and Astronomy, Beijing Normal University, Beijing 100875, China}
\affiliation{Key Laboratory of Multiscale Spin Physics (Ministry of Education), Beijing Normal University, Beijing 100875, China}

\begin{abstract}
In all known schemes for the non-Hermitian skin effect, the non-Hermitian ingredient that drives the skin localization, whether asymmetric hopping or gain and loss, is invariably introduced by hand as an independent model parameter along the skin direction. Here we show that when two spatial coordinates do not commute, the skin effect can break free of this paradigm: a gain-loss potential applied along one coordinate automatically generates non-reciprocity along the other through the coordinate noncommutativity, driving all eigenstates to pile up exponentially at a boundary. We term this phenomenon the noncommutative skin effect. The inverse skin length is proportional to the noncommutativity parameter and is given by an analytic formula, exact in the thermodynamic limit and verified by exact diagonalization of lattice models; the reflection symmetry of the imaginary potential furnishes an exact criterion for the presence or absence of the effect, valid rigorously for finite-size systems. For a sinusoidal imaginary potential, the skin direction of all eigenstates flips collectively at parameter points fixed purely by geometry. Because the flip point is independent of the potential strength, the reversal constitutes a zero-crossing measurement scheme intrinsically robust against systematic errors, from which the noncommutativity parameter can be extracted directly. The qualitative transition of the eigenstates from uniform to exponentially localized renders the effect a nonperturbative probe of spatial noncommutativity, and the Peierls-phase structure of its lattice model is in principle accessible to cold-atom synthetic dimensions, photonic resonators, and topolectrical circuits.
\end{abstract}

\maketitle

\section{Introduction}\label{sec:intro}

The non-Hermitian skin effect (NHSE), the exponential boundary accumulation of an extensive number of eigenstates under open boundary conditions (OBC) together with the accompanying spectral collapse~\cite{Yao2018,Yokomizo2019}, has overturned the Bloch bulk-boundary correspondence~\cite{Yao2018,zhang_CorrespondenceWindingNumbers_2020,borgnia_NonHermitianBoundaryModes_2020,yang_NonHermitianBulkBoundaryCorrespondence_2020,kawabata_SymmetryTopologyNonHermitian_2019}, given rise to the non-Bloch band theory~\cite{Yokomizo2019,wang_AmoebaFormulationNonBloch_2024}, and been observed on photonic, electrical-circuit, acoustic, cold-atom, and quantum-walk platforms~\cite{weidemann_TopologicalFunnelingLight_2020,helbig_GeneralizedBulkBoundary_2020,zhang_AcousticNonHermitianSkin_2021,Liang2022,Zhao2025,xiao_NonHermitianBulkBoundary_2020} (for reviews, see Refs.~\cite{ashida_NonHermitianPhysics_2020,Ding2022review,Okuma2023review}). Its variants have grown remarkably diverse, including the critical skin effect~\cite{li_CriticalNonHermitianSkin_2020}, scale-free localization~\cite{li_ImpurityInducedScalefree_2021,li_ScalefreeLocalizationPT_2023}, higher-order skin effects~\cite{zhang_HigherOrderNHSE_2021}, geometry-dependent skin effects~\cite{zhang_UniversalNonHermitianSkin_2022,Wang2023GDSE,Zhou2023GDSE}, boundary-insensitive skin effects~\cite{Wei2025}, and skin-direction reversal~\cite{Li2022reversal,Yuce2024bipolar,Yang2026transverseSwitch}; the driving mechanisms encompass asymmetric hopping~\cite{Hatano1996}, gain and loss~\cite{yi_NonHermitianSkinModes_2020,jiang_TunableNonHermitianSkin_2024}, dissipative coupling~\cite{Huang2024dissipative}, Floquet driving~\cite{Lin2024Floquet,Sun2024Floquet}, geometry~\cite{Wang2023GDSE}, and spatially inhomogeneous imaginary potentials~\cite{Wei2025}. Yet one feature has remained common to all of them: the non-Hermitian ingredient that drives the NHSE, be it asymmetric hopping, gain and loss, or imaginary couplings, always enters as an \emph{independent model input along the skin direction}; even in schemes that dispense with asymmetric hopping, non-Hermitian parameters still pervade the skin direction. To date, no scheme exists in which the non-reciprocity along the skin direction is not an independent input but is instead dictated by the fundamental structure of space itself (see Table~\ref{tab:NCSE-vs-HN} for an item-by-item comparison with existing mechanisms).

Coordinate noncommutativity, $[\hat{x},\hat{y}]=i\theta$, offers a natural way out. This algebraic structure arises in a variety of fundamental settings, including the guiding centers of Landau levels~\cite{Bellissard1994,Ezawa2013}, D-branes in string theory~\cite{Seiberg1999}, and minimal-length scenarios in quantum gravity~\cite{Connes1994,Douglas2001}. In recent years, synthetic-dimension platforms have independently accomplished two things: the simulation of algebraic structures related to coordinate noncommutativity (non-Abelian gauge fields~\cite{Cheng2025} and Landau-level physics~\cite{Zhang2024CLM}) and the observation of the NHSE~\cite{Liang2022,Zhao2025,Wang2021winding,Zheng2024}. The theoretical connection between the two, however, has long remained at the level of spectral corrections and $PT$ symmetry~\cite{Giri2009,Santos2019}, lacking a deeper point of contact. Within noncommutative quantum mechanics, existing proposals for detecting $\theta$ rely on perturbative energy corrections~\cite{Gamboa2001}, whose signal vanishes continuously as $\theta\to 0$ without any macroscopic qualitative jump. By contrast, the spectral and wave-function reconstruction underlying the NHSE is intrinsically nonperturbative~\cite{Budich2020,Guo2021,Okuma2020}: the change of the density profile from uniform to exponentially localized cannot be undone by any continuous deformation. A natural question therefore arises: can the non-reciprocity along the skin direction be generated with coordinate noncommutativity as the sole transmission mechanism, rather than being introduced as an independent model parameter? If the answer is affirmative, the NHSE becomes a macroscopic, qualitative probe of spatial noncommutativity, with a sensitivity fundamentally beyond that of perturbative corrections.

In this work we show that the answer is affirmative. The mechanism can be summarized as a three-step chain (Fig.~\ref{fig:mechanism}). First, representation theory locks in a shear structure: the Bopp shift $\hat{y}=\hat{Y}+\theta\hat{P}_X$~\cite{Douglas2001,Szabo2003} is the unique faithful representation of the noncommutative algebra~\eqref{eq:commutation} (Stone--von Neumann theorem~\cite{Gouba2009,Scholtz2009}; Theorem~\ref{thm:no-go}), so an imaginary potential along $\hat{y}$ unavoidably carries the momentum of the $X$ direction: a mode with momentum $k$ experiences the imaginary-potential window translated by $\theta k$ along $y$ [Fig.~\ref{fig:mechanism}(a)]. Second, the translation produces dispersion non-reciprocity: under open boundaries along $Y$, right-moving ($k>0$) and left-moving ($k<0$) modes sample windows with different average gain and loss, so the imaginary part of the spectrum becomes an asymmetric function of $k$ [Fig.~\ref{fig:mechanism}(b); Sec.~\ref{sec:model}]. Third, the dispersion non-reciprocity converts into a skin effect: once open boundaries are imposed along $X$ as well, the momentum-dependent gain imbalance can no longer be carried by Bloch waves, and the same mechanism as in the Hatano--Nelson model~\cite{Hatano1996,Yao2018} drives all eigenstates to accumulate collectively, with exponential envelopes, on the gain-favored side (Sec.~\ref{sec:quantitative}). Throughout this chain the only non-Hermitian input is the imaginary potential along $\hat{y}$; the non-reciprocity along the skin direction is uniquely locked by the noncommutative algebra together with that potential and cannot be tuned independently, a fundamental departure from conventional schemes, where the non-Hermitian parameters are free inputs. We name this effect the \textbf{noncommutative skin effect} (NCSE); its mechanism and the main results of this work are summarized in Fig.~\ref{fig:mechanism}.

The properties of the NCSE are characterized by two complementary theorems. Theorem~\ref{thm:NHSE} [Sec.~\ref{sec:quantitative}; Fig.~\ref{fig:mechanism}(b)] provides a quantitative formula for the inverse skin length, $\kappa_\nu\propto\theta$ (exact in the thermodynamic limit), so that the strength and direction of the skin effect directly encode the noncommutative structure of the underlying space. Theorem~\ref{thm:symmetry} [Sec.~\ref{sec:symmetry}; Fig.~\ref{fig:mechanism}(c)] establishes the reflection symmetry of the imaginary potential as an exact criterion for the existence of the NCSE, valid rigorously for finite systems without any approximation. The symmetry criterion further predicts a directly measurable phenomenon (Sec.~\ref{sec:peierls}): for a sinusoidal imaginary potential, the system size $\Ly$ along $Y$ controls the skin direction along $X$, and all modes reverse collectively at the flip points $\theta_c=\Ly/[(2k+1)\pi]$, which are fixed solely by the geometric parameter $\omega=\Ly/\ell$; the noncommutativity parameter $\theta$ can thus be extracted directly, without any knowledge of the potential strength $g$. This strategy, taking the zero crossing of a direction reversal rather than the absolute strength of a signal as the observable, is intrinsically robust against systematic errors, much as locating the dark fringes of an interferometer is far more precise than measuring absolute intensities. All of the analytic predictions above have been verified numerically (Secs.~\ref{sec:numerical-verification}, \ref{sec:symmetry-robustness}, and \ref{sec:peierls}; Figs.~\ref{fig:quantitative-verification}--\ref{fig:periodic-verification}). Since the direction-flip prediction relies only on symmetry rather than on implementation details, it can in principle be tested on cold-atom synthetic dimensions~\cite{Liang2022,Zhao2025}, photonic resonators~\cite{Wang2021winding,Zheng2024}, and topolectrical circuits~\cite{helbig_GeneralizedBulkBoundary_2020} (Sec.~\ref{sec:discussion}), providing a sharp and robust target for experimental verification.

\begin{figure*}[t]
\centering
{\includegraphics[width=1\linewidth]{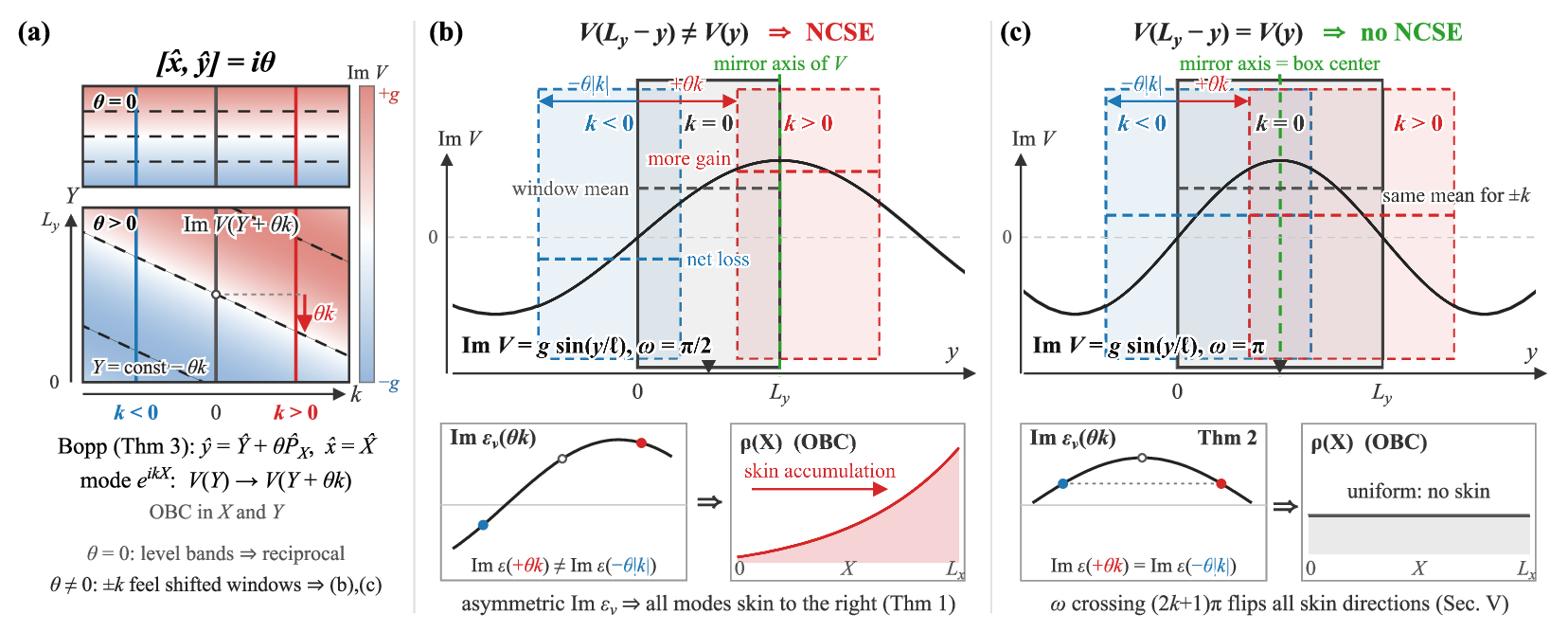}}
\caption{Sliding-window mechanism of the NCSE and roadmap of this work. (a)~Shear picture: both canvases show the $(k,Y)$ plane (horizontal axis: wave vector $k$ along $X$; vertical axis: transverse coordinate $Y\in[0,\Ly]$), and the color at point $(k,Y)$ encodes the imaginary potential $\mathrm{Im}\,V(Y+\theta k)$ experienced at $Y$ by the mode with momentum $k$ (red: gain; blue: loss; see the color bar on the right). Top panel, $\theta=0$: the contour bands are horizontal (black dashed lines mark the contours), and the vertical cuts at any fixed $k$ (blue/gray/red vertical lines) yield identical transverse potential profiles; the dispersion is reciprocal. Bottom panel, $\theta>0$: the Bopp shift $\hat{y}=\hat{Y}+\theta\hat{P}_X$ shears momentum into coordinate, and the contours become the family of straight lines $Y=\mathrm{const}-\theta k$ with slope $-\theta$; the vertical cut at fixed $k$ is the potential profile experienced by the mode $e^{ikX}$, translated farther along $Y$ for larger momentum (red arrow: the displacement $\theta k$ of one contour band from the $k=0$ cut to a $k>0$ cut). The $\pm k$ cuts are precisely the windows sampled by the sliding boxes in (b) and (c). Representation uniqueness (Theorem~\ref{thm:no-go}) guarantees that the shear is an algebraic necessity rather than a modeling choice. (b)~Mirror-asymmetric placement [sinusoidal potential $\mathrm{Im}\,V=g\sin(y/\ell)$ with $\omega\equiv\Ly/\ell=\pi/2$] $\Rightarrow$ NCSE: $\epnu(s)$ is the $\nu$th eigenvalue of the Dirichlet box $[s,\,s+\Ly]$ ($s=\theta k$) sliding across the fixed potential landscape; the black triangles mark the box center, and the green dashed line is the mirror axis of the potential (offset from the box center). The $k>0$ (red) and $k<0$ (blue) windows sample different window-averaged gains (dashed horizontal lines, computed from the actual values), so $\mathrm{Im}\,\epnu(\theta k)$ is an asymmetric function of $k$ (lower-left inset, actual curve), with a nonzero slope at $s=0$; Theorem~\ref{thm:NHSE} then yields the inverse skin length $\kappa_\nu=(\theta/2)\,\partial_s\mathrm{Im}\,\epnu|_{s=0}\neq 0$, and under OBC the $X$ density $\rho(X)$ of all eigenstates piles up on the $X=\Lx$ side with this common inverse length (lower-right inset). (c)~Mirror-symmetric placement (same sinusoidal potential, $\omega=\pi$) $\Rightarrow$ no NCSE: the mirror axis passes exactly through the box center (green dashed line coincides with the triangles), the $\pm k$ windows are mirror images of each other with equal window-averaged gain, $\epnu(s)$ is an even function (Theorem~\ref{thm:symmetry}), and $\rho(X)$ is uniform. Tuning $\omega$ so that the mirror axis sweeps through the box center [$\omega$ crossing $(2k+1)\pi$] flips the skin direction of all modes collectively (Sec.~\ref{sec:peierls}).}
\label{fig:mechanism}
\end{figure*}

\section{Model and mechanism}\label{sec:model}
We begin with quantum mechanics on a two-dimensional noncommutative space, where the coordinate operators $\hat{x}$, $\hat{y}$ and their conjugate momenta $\hat{p}_x$, $\hat{p}_y$ obey the commutation relations
\begin{equation}\label{eq:commutation}
[\hat{x},\hat{y}]=i\theta,\quad [\hat{x},\hat{p}_x]=[\hat{y},\hat{p}_y]=i,
\end{equation}
with $[\hat{p}_x,\hat{p}_y]=[\hat{x},\hat{p}_y]=[\hat{y},\hat{p}_x]=0$, where $\theta$ is the noncommutativity parameter; $\theta=0$ recovers the standard commutative space. The full algebraic structure of Eq.~\eqref{eq:commutation} admits no faithful representation on a single $L^2(\mathbb{R})$ (see Appendix~\ref{app:dim-obstruction} for the proof); the minimal faithful representation requires $L^2(\mathbb{R}^2)$. To this end we introduce a set of standard canonical variables $(\hat{X},\hat{Y},\hat{P}_X,\hat{P}_Y)$ with $[\hat{X},\hat{P}_X]=[\hat{Y},\hat{P}_Y]=i$ and all other commutators vanishing, and define the asymmetric Bopp shift~\cite{Douglas2001,Szabo2003}:
\begin{equation}\label{eq:Bopp}
\hat{x}=\hat{X},\quad \hat{y}=\hat{Y}+\theta\hat{P}_X,\quad \hat{p}_x=\hat{P}_X,\quad \hat{p}_y=\hat{P}_Y.
\end{equation}
The advantage of the asymmetric form lies in $\hat{x}=\hat{X}$: the wave-function profile along $X$ directly reflects the distribution along the physical coordinate $\hat{x}$, which renders the analysis of the skin effect most transparent (the symmetric Bopp shift and its equivalence are discussed in Appendix~\ref{app:Bopp-star}). We consider a free particle subject to an arbitrary imaginary potential $V(\hat{y})=iU(\hat{y})$ along $\hat{y}$, with $U$ a real-valued function (setting $\hbar=2m=1$):
\begin{equation}\label{eq:H-original}
\hat{H}=\hat{p}_x^2+\hat{p}_y^2+V(\hat{y}).
\end{equation}
Under the Bopp shift, the Hamiltonian is mapped onto commutative space:
\begin{equation}\label{eq:Bopp-H}
\hat{H}_{Bop}=\hat{P}_X^2+\hat{P}_Y^2+V\!\left(\hat{Y}+\theta\hat{P}_X\right).
\end{equation}
In the coordinate representation ($\hat{P}_X=-i\partial_X$, $\hat{P}_Y=-i\partial_Y$), the eigenvalue problem $\hat{H}_{Bop}\Psi=E\Psi$ reads
\begin{equation}\label{eq:Bopp-H-PDE}
\left[-\left(\partial_X^2+\partial_Y^2\right)+V\!\left(Y-i\theta\partial_X\right)\right]\Psi=E\Psi.
\end{equation}
For $V(\hat{y})$, since $[\hat{Y},\hat{P}_X]=0$ (they act on different degrees of freedom), the potential term in Eq.~\eqref{eq:Bopp-H} can be Taylor expanded free of any operator-ordering ambiguity:
\begin{equation}\label{eq:Bopp-expansion}
V(\hat{Y}+\theta\hat{P}_X)=V(\hat{Y})+V'(\hat{Y})\theta\hat{P}_X+\tfrac{V''(\hat{Y})(\theta\hat{P}_X)^2}{2}+\cdots,
\end{equation}
where the zeroth-order term $V(\hat{Y})$ is the potential along $Y$, the first-order term $U'(Y)\theta\partial_X$ is, in the coordinate representation, an imaginary gauge field~\cite{Hatano1996} along $X$ whose strength is proportional to $\theta$ and to the local potential gradient, and the even orders produce only symmetric corrections (see the lattice analysis in Appendix~\ref{app:lattice-general}). The Bopp shift therefore forces the imaginary potential along $\hat{y}$ to leak into the $X$ dynamics, and the leading term $U'(Y)\theta\partial_X$ possesses precisely the imaginary-gauge-field structure required to drive the NCSE. Moreover, since Eq.~\eqref{eq:Bopp-H} does not contain $\hat{X}$ explicitly, $\hat{P}_X$ is conserved. Taking eigenstates of $\hat{P}_X$, $\Psi(X,Y)=e^{ikX}\phi_\nu(Y;s)$ with $s\equiv\theta k$, where $k$ is the good quantum number of $\hat{P}_X$ and $\nu=1,2,\ldots$ labels the bands, the potential becomes $V(Y+s)$, i.e., translated by $s$ along $Y$, a statement that holds exactly for arbitrary functional forms of $V$. We then have
\begin{equation}\label{eq:h-def}
h(s)\,\phi_\nu(Y;s)=\epnu(s)\,\phi_\nu(Y;s),
\end{equation}
where $h(s)\equiv -\frac{d^2}{dY^2}+V(Y+s)$ is the effective Hamiltonian along $Y$, subject to OBC on $Y\in[0,\Ly]$ [$\phi_\nu(0;s)=\phi_\nu(\Ly;s)=0$]. The total energy is $E_\nu(k)=k^2+\epnu(\theta k)$, with the potential $V(Y+s)$ depending on the wave vector $k$ through $s$. Here $\epnu(s)$ is nothing but the $\nu$th eigenvalue of the Dirichlet box $[s,\,s+\Ly]$ in a fixed potential landscape: a particle confined to a fixed box samples a sliding window of the global landscape, and the transverse boxes of right- and left-moving modes are translated in opposite directions by $\theta|k|$ along $y$ [Fig.~\ref{fig:mechanism}(b)]. Whenever $\epnu(s)\neq\epnu(-s)$, one has $E_\nu(k)\neq E_\nu(-k)$: the dispersion is non-reciprocal. This window picture can be visualized globally in the $(k,Y)$ plane [Fig.~\ref{fig:mechanism}(a)]: $V(Y+\theta k)$ depends on the two variables only through the combination $Y+\theta k$, so its contours form a family of straight lines with slope $-\theta$; the noncommutativity shears momentum into coordinate. At $\theta=0$ the contour bands are horizontal and all modes experience the same transverse potential; at $\theta\neq 0$ the vertical cut at fixed $k$ gives the transverse potential profile of that mode, and different cuts are translates of one another. The shear itself is merely a kinematic structure; converting it into spectral non-reciprocity requires the OBC along $Y$ to break the unitary equivalence between the $\pm s$ windows, which is the starting point of the quantitative theory developed in the next section (Sec.~\ref{sec:perturbative}; Appendix~\ref{app:BC-necessity}).

\subsection{Exactly solvable case: linear imaginary potential}\label{sec:linear-example}
We take the linear imaginary potential as an example to show how an imaginary potential on a noncommutative space naturally gives rise to the NCSE. Substituting $V=ig\hat{y}$ into Eq.~\eqref{eq:h-def}, the effective Hamiltonian along $Y$ is $h(s)=-d^2/dY^2+ig(Y+s)$. The constant piece $igs$ does not modify $\phi_\nu$ but merely shifts the eigenvalue, $\epnu(s)=\epnu(0)+igs$, where $\epnu(0)$ is the $\nu$th eigenvalue of $h(0)$; the total energy is therefore $E_\nu(k)=k^2+\epnu(0)+ig\theta k$.

Since both $\phi_\nu(Y;s=0)$ and $\epnu(0)$ are independent of $k$, the $X$ and $Y$ directions decouple completely and the total wave function separates as $\Psi(X,Y)=\psi(X)\cdot\phi_\nu(Y)$. In the coordinate representation ($k\to -i\partial_X$), $\psi(X)$ obeys the $X$-direction eigenvalue equation $(-\partial_X^2+g\theta\,\partial_X)\psi=E_X\psi$, where the imaginary gauge field $g\theta\partial_X$ originates purely from the noncommutativity.

Under periodic boundary conditions (PBC) on $X\in[0,\Lx]$, $\psi(0)=\psi(\Lx)$, the Bloch solutions $\psi_k(X)=e^{ikX}$ are eigenstates of $\hat{P}_X$ with $E_X^{\mathrm{PBC}}(k)=k^2+ig\theta k$, whose imaginary part $g\theta k$ is an odd function of $k$ (non-reciprocal dispersion): right- and left-moving modes experience different growth (decay) rates, and the PBC spectrum traces out a parabola in the complex plane. Under OBC, $\psi(0)=\psi(\Lx)=0$, translation invariance is broken and $\hat{P}_X$ is no longer conserved, so the OBC eigenstates are not eigenstates of $\hat{P}_X$. Writing $\psi(X)=e^{\kappa X}\lambda(X)$ and inserting it into the $X$-direction eigenvalue equation, the coefficient of $\partial_X\lambda$ is $-2\kappa+g\theta$; setting it to zero yields $\kappa=g\theta/2$, and the remaining equation is a standard particle-in-a-box problem, $\lambda''= (g^2\theta^2/4-E_X)\lambda$. The OBC eigenstates are
\begin{equation}\label{eq:linear-OBC-state}
\psi_n(X)\propto e^{g\theta X/2}\sin\frac{n\pi X}{\Lx},\quad n=1,2,\ldots
\end{equation}
with the corresponding OBC energies $E_{X,n}^{\mathrm{OBC}}=(n\pi/\Lx)^2+g^2\theta^2/4$ being purely real: the OBC spectrum collapses onto the real axis while the PBC spectrum covers a parabola in the complex plane, the spectral fingerprint of the NCSE. All states share one and the same exponential envelope, with inverse skin length $\kappa=g\theta/2$ (we adopt the convention that $\kappa>0$ corresponds to localization at $X=\Lx$ and $\kappa<0$ to localization at $X=0$). The $Y$-direction Hamiltonian $H_Y=-\partial_Y^2+igY$ exhibits a non-Hermitian localization effect of its own~\cite{Wei2025}, but here we focus on the NCSE along $X$, which is strictly absent in the standard commutative space ($\theta=0$). Exact lattice diagonalization confirms these analytic results: Fig.~\ref{fig:quantitative-verification}(a) shows the eigenstate-averaged density $\bar\rho(m)$ along $X$ for the linear imaginary potential on a semilogarithmic scale; the numerical data fall precisely on the analytic envelope $\propto e^{g\theta m}$, with the inverse skin length $\kappa=g\theta/2$ independent of the band index $\nu$.

\section{Quantitative theory of the NCSE}\label{sec:quantitative}
The analysis of the linear potential reveals the basic mechanism of the NCSE, yet that case is special in ways that do not generalize: $V''=0$ truncates the Bopp expansion exactly at first order, the two directions decouple completely, and $\kappa=g\theta/2$ holds exactly for arbitrary $g$. For a general imaginary potential, the imaginary-gauge-field strength $U'(Y)$ varies with position and the two directions no longer separate. In this section we develop a quantitative theory of the NCSE for arbitrary imaginary potentials.

\subsection{Non-reciprocity coefficient}\label{sec:perturbative}

For a general imaginary potential $V(\hat{y})=iU(\hat{y})$, the non-reciprocal part of the dispersion $E_\nu(k)=k^2+\epnu(s)$, namely $f_\nu(s)\equiv \epnu(s)-\epnu(-s)$, is an odd function of $s$ whose leading behavior $f_\nu(s)\approx 2\epnu'(0)s$ is controlled by the parametric derivative $\epnu'(0)$; the latter can be computed exactly via the non-Hermitian Hellmann--Feynman theorem~\cite{Moiseyev2011,Hajong2023} (full derivation in Appendix~\ref{app:alpha-derivation}). The inequality $\epnu(s)\neq\epnu(-s)$ means that right-moving ($k>0$) and left-moving ($k<0$) modes experience different growth or decay rates: their transverse potential windows are shifted by $\pm\theta|k|$ and sample different window-averaged gains [Fig.~\ref{fig:mechanism}(b)]. This also explains why the $Y$ direction in Eq.~\eqref{eq:h-def} must be a finite interval with open boundaries: in the absence of boundaries, or under PBC, the translation operator $T_s:\phi(Y)\mapsto\phi(Y+s)$ is always unitary, $h(s)$ is unitarily equivalent to $h(0)$, and the spectrum is independent of $s$; only OBC break this unitary equivalence and make $\epnu(s)\neq\epnu(-s)$ possible (a rigorous proof and a systematic classification of boundary conditions are given in Appendix~\ref{app:BC-necessity}). Converting the dispersion non-reciprocity into one-sided accumulation of the eigenstates, in turn, requires OBC along $X$~\cite{Yao2018} (Sec.~\ref{sec:adiabatic}); the two sets of OBC thus play distinct and indispensable roles in producing the NCSE.
\begin{proposition}[Non-reciprocity coefficient]\label{prop:alpha}
Let $\phi_\nu^R$ be the $\nu$th right eigenstate of $h(0)$. For a purely imaginary potential $V(\hat{y})=iU(\hat{y})$ ($U$ real valued),
\begin{equation}\label{eq:alpha}
\epnu'(0)=i\alpha_\nu,\; \alpha_\nu\equiv\frac{\displaystyle\int_0^{\Ly}[\phi_\nu^R(Y)]^2\,U'(Y)\,dY}{\displaystyle\int_0^{\Ly}[\phi_\nu^R(Y)]^2\,dY}\,.
\end{equation}
In general $\alpha_\nu\in\mathbb{C}$, since $\phi_\nu^R$ is a complex function and $[\phi_\nu^R]^2\neq|\phi_\nu^R|^2$.
\end{proposition}
The appearance of the bilinear form $[\phi_\nu^R(Y)]^2$ in the integrand, rather than the modulus squared $|\phi_\nu^R(Y)|^2$, stems from the biorthogonal structure of non-Hermitian quantum mechanics: for a purely imaginary potential, $h(0)^\dagger=T-iU$ ($T=-d^2/dY^2$), so the left eigenstates are $\phi_\nu^L=(\phi_\nu^R)^*$; in the biorthogonal inner product $(\phi_\nu^L)^*=\phi_\nu^R$, which produces $[\phi_\nu^R]^2$.
When $|s|=|\theta k|$ is small compared with the scale $s_*$ over which $\epnu(s)$ varies nonlinearly (i.e., the scale at which the second-order term becomes comparable to the first-order one, $s_*\sim|\epnu'(0)|/|\epnu''(0)|$; non-Hermitian perturbation theory estimates its magnitude as the ratio of the band gap to the potential gradient, $s_*\sim|\epnu-\epsilon_\mu|/\|V'\|\equiv s_\mathrm{gap}$, see Appendix~\ref{app:adiabatic}), the Taylor truncation $\epnu(s)\approx\epnu(0)+i\alpha_\nu s$ is valid, and inserting it into the dispersion relation yields the effective PBC dispersion
\begin{equation}\label{eq:dispersion-PBC}
E_\nu(k)\approx k^2+\epnu(0)+i\alpha_\nu\theta k.
\end{equation}
Its imaginary part, $\mathrm{Im}[E_\nu(k)]=\mathrm{Im}[\epnu(0)]+\mathrm{Re}[\alpha_\nu]\theta k$, is an odd function of $k$, the direct manifestation of dispersion non-reciprocity.

As an odd function of $s$, the non-reciprocal function $f_\nu(s)$ satisfies $f_\nu(0)=0$ and its Taylor series contains only odd powers:
\begin{equation}\label{eq:f-Taylor}
f_\nu(s)=2\epnu'(0)\,s+\tfrac{1}{3}\epnu'''(0)\,s^3+\cdots
\end{equation}
(the analyticity of $\epnu(s)$ at $s=0$ is guaranteed by the nondegeneracy condition~\cite{Kato1966}). Dispersion non-reciprocity ($f_\nu\not\equiv 0$) therefore implies the existence of a smallest odd order $2p+1$ with $\epnu^{(2p+1)}(0)\neq 0$. Since the nongeneric cases $p\geq 1$ require fine-tuning, we focus for convenience on the generic case $p=0$ [$\epnu'(0)\neq 0$].

\subsection{Effective Hatano--Nelson equation}\label{sec:adiabatic}
In Eq.~\eqref{eq:dispersion-PBC}, $k$ is a real good quantum number under PBC along $X$. Under OBC, however, $\hat{P}_X$ is no longer conserved and the dispersion relation cannot be used directly. Promoting the PBC dispersion to an effective OBC equation requires two approximations: (i)~an adiabatic approximation that reduces the two-dimensional OBC problem to band-resolved one-dimensional equations, and (ii)~a Taylor truncation that converts the exact one-dimensional equation into a solvable Hatano--Nelson form. Both become automatically exact in the thermodynamic limit $\Lx\to\infty$.

For a given energy $E$ and band index $\nu$, the $X$-direction wave vector $\beta$ is determined by the dispersion relation $\beta^2+\epnu(\theta\beta)=E$. Since $\epnu(s)$ is a general analytic function of $s$ (not a polynomial), this is a transcendental equation in $\beta$, generally possessing several (even infinitely many) roots $\beta_1,\beta_2,\ldots$\,. Each root corresponds to a bulk solution $e^{i\beta_j X}\phi_\nu(Y;\theta\beta_j)$, and the OBC eigenstate is their linear superposition:
\begin{equation}\label{eq:Psi-general}
\Psi(X,Y)=\sum_j c_j\,e^{i\beta_j X}\,\phi_\nu(Y;\theta\beta_j).
\end{equation}
The OBC along $X$ demand that $\Psi(0,Y)=0$ and $\Psi(\Lx,Y)=0$ hold simultaneously for all $Y\in[0,\Ly]$. At $X=0$,
\begin{equation}\label{eq:BC-functional}
\sum_j c_j\,\phi_\nu(Y;\theta\beta_j)=0,\qquad\forall\,Y\in[0,\Ly].
\end{equation}
This is a \emph{functional equation}: it is not a finite set of numerical conditions on the coefficients $\{c_j\}$, but rather requires a function of the continuous variable $Y$ to vanish identically. Expanding the $Y$ dependence in a complete basis [e.g., $\sin(m\pi Y/\Ly)$], Eq.~\eqref{eq:BC-functional} is equivalent to the vanishing of every expansion coefficient, i.e., infinitely many constraints imposed on finitely many unknowns $\{c_j\}$, an overdetermined system that generically admits no solution. We therefore invoke the adiabatic approximation, assuming that the variation of $\phi_\nu(Y;\theta\beta_j)$ among different roots is negligible: $\phi_\nu(Y;\theta\beta_j)\approx\phi_\nu(Y;0)$ for all $j$. The transverse wave function can then be factored out of the sum:
\begin{equation}\label{eq:Psi-factored}
\Psi(X,Y)\approx\phi_\nu(Y;0)\cdot\psi(X),\;\psi(X)\equiv\sum_j c_j e^{i\beta_j X},
\end{equation}
and the functional equation collapses to the scalar conditions $\psi(0)=\psi(\Lx)=0$. The adiabatic condition requires the parameter spacing of adjacent roots, $\Delta s=\theta|\beta_+-\beta_-|$, to be small compared with the characteristic scale $s_{\mathrm{gap}}\sim|\epnu-\epsilon_\mu|/\|V'\|$ over which the transverse wave function changes ($\epsilon_\mu$ is the eigenvalue of the nearest band and $\|V'\|\equiv\max_Y|V'(Y)|$; see Appendix~\ref{app:adiabatic} for the derivation). For a fixed mode index $n$, $\Delta s=2\pi n\theta/\Lx\propto 1/\Lx$ vanishes in the thermodynamic limit, and the condition is satisfied automatically. The $\{\phi_\nu(Y;0)\}$ of different bands are linearly independent, so the $X$-direction problem can be solved independently for each band.

Under the adiabatic condition, each plane-wave component $e^{i\beta_j X}$ of $\psi(X)=\sum_j c_j e^{i\beta_j X}$ obeys the dispersion relation $\beta_j^2+\epnu(\theta\beta_j)=E$. Since $e^{i\beta X}$ is an eigenfunction of the operator $-i\partial_X$ (with eigenvalue $\beta$), for any analytic function $\epnu$ one has $\epnu(-i\theta\partial_X)e^{i\beta X}=\epnu(\theta\beta)e^{i\beta X}$; superposing these components, $\psi(X)$ satisfies the effective one-dimensional equation [$\psi(0)=\psi(\Lx)=0$]
\begin{equation}\label{eq:exact-1D}
\begin{split}
\left[-\partial_X^2+\epnu(-i\theta\partial_X)\right]\psi=E\,\psi.
\end{split}
\end{equation}
This equation does not rely on any Taylor expansion and holds for arbitrary $\epnu(s)$. The operator $\epnu(-i\theta\partial_X)$ is a pseudo-differential operator~\cite{Zworski2012}: its symbol $\epnu(\theta\beta)$ is a general analytic function of $\beta$ rather than a polynomial, corresponding in position space to a nonlocal convolution operator (see Appendix~\ref{app:exact-1D}). Equation~\eqref{eq:exact-1D} is therefore exact, but the nonlocal operator $\epnu(-i\theta\partial_X)$ precludes an analytic solution; at this point we resort to the Taylor truncation, which reduces it to an algebraically solvable form. Using $\epnu(s)\approx\epnu(0)+i\alpha_\nu\,s$ (Proposition~\ref{prop:alpha}), Eq.~\eqref{eq:exact-1D} becomes the effective Hatano--Nelson equation
\begin{equation}\label{eq:HN}
\left[-\partial_X^2+\alpha_\nu\theta\partial_X+\epnu(0)\right]\psi(X)=E\psi(X).
\end{equation}
This is precisely the structure obtained by promoting the PBC dispersion Eq.~\eqref{eq:dispersion-PBC} of Sec.~\ref{sec:perturbative} to the coordinate representation ($k\to -i\partial_X$), here endowed with a rigorous OBC foundation through the adiabatic reduction. The purpose of the truncation is not to ``turn'' the non-polynomial $\epnu(s)$ into a polynomial, but to obtain analytic solutions while retaining the essential physics: the qualitative conclusion (dispersion non-reciprocity implies the NCSE) is guaranteed by the exact symmetry theorem of Sec.~\ref{sec:symmetry} and does not rely on the truncation; only the quantitative formulas do.

\begin{figure}[ptb]
\centering
{\includegraphics[width=1\linewidth]{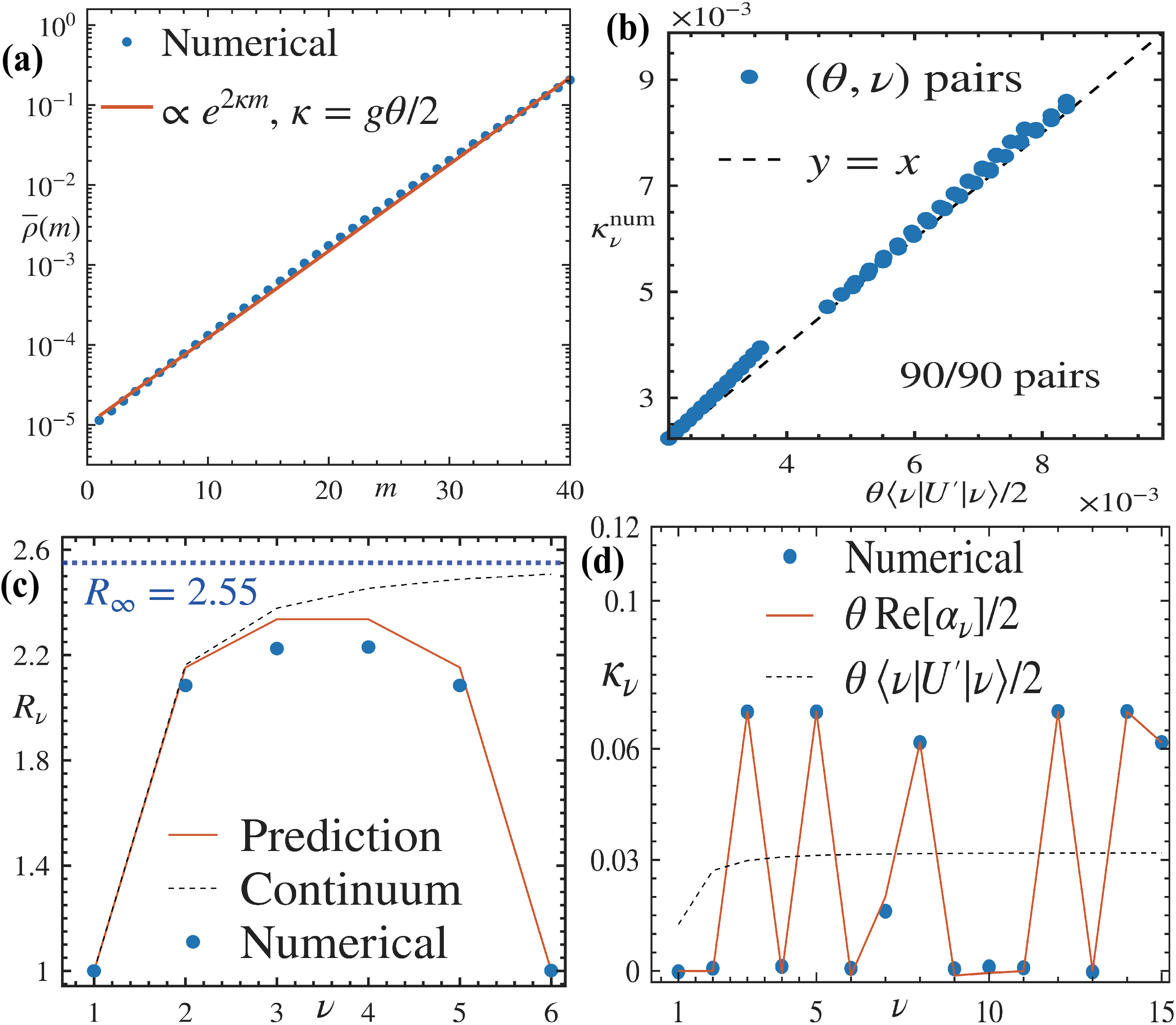}}
\caption{Numerical verification of Theorem~\ref{thm:NHSE}.
(a)~Linear imaginary potential $V{=}ig\hat{y}$ ($N_x{=}40$, $N_y{=}30$, $g{=}0.05$, $\theta{=}5$). The vertical axis is the eigenstate-averaged density along $X$, $\bar{\rho}(m){\equiv}\frac{1}{N_\mathrm{eig}}\sum_{\alpha,n}|\psi_\alpha(m,n)|^2$ (semilogarithmic scale), and the horizontal axis is the $X$-direction site index $m$; blue circles, numerical results; red solid line, analytic envelope $\propto e^{g\theta m}$. The two agree precisely.
(b),(c)~Numerical verification for the cubic imaginary potential $V{=}ig(Y{-}L_y/2)^3$ in the weak-potential limit ($N_x{=}500$, $N_y{=}6$, $g{=}3{\times}10^{-3}$, 15 values of $\theta\in[0.30,\,0.50]$):
(b)~the horizontal axis is Eq.~\eqref{eq:kappa-weak} and the vertical axis is $\kappa_\nu^{\mathrm{num}}$ extracted band by band from exact diagonalization. Each blue data point corresponds to one $(\theta,\nu)$ pair ($90=15$ values of $\theta$ $\times$ $N_y$ bands in total; each $\kappa_\nu$ is obtained by fitting the densities of the $N_x$ eigenstates within that band). The data points falling on the black dashed line $y{=}x$ confirm Eq.~\eqref{eq:kappa-weak}.
(c)~The vertical axis is the ratio $R_\nu{\equiv}\kappa_\nu/\kappa_1$ and the horizontal axis is the band index $\nu$: blue circles, numerical results (for each of the 15 values of $\theta\in[0.30,\,0.50]$, $R_\nu$ is computed separately and the median is taken, robustly removing extraction outliers at individual parameters); red solid line, analytic lattice prediction for $R_\nu$; black dashed line, analytic continuum-limit prediction; blue dotted line, asymptotic value $R_\infty{\approx}2.55$.
(d)~Numerical and analytic $\kappa_\nu$ for the cubic potential at general potential strength, with $N_x{=}60$, $N_y{=}15$, $g{=}0.01$, $\theta{=}0.10$; the horizontal axis is the band index $\nu$. Blue circles, numerical results; red solid line, analytic prediction of Theorem~\ref{thm:NHSE} [with $\alpha_\nu$ computed from the numerical eigenstates of $h(0)$]; black dashed line, weak-potential analytic prediction. Theorem~\ref{thm:NHSE} tracks the oscillatory structure, whereas the weak-potential formula deviates completely.
}
\label{fig:quantitative-verification}
\end{figure}

\subsection{Skin states, spectra, and the unified criterion}\label{sec:NCSE-derivation}
We now solve Eq.~\eqref{eq:HN} explicitly to establish the state-level and spectral signatures of the NCSE. Decomposing $\alpha_\nu=\aR+i\aI$ into its real and imaginary parts and setting $\psi=e^{ikX}$, the characteristic equation of Eq.~\eqref{eq:HN} reads
\begin{equation}\label{eq:char-eq}
k^2+i\alpha_\nu\theta\,k=E-\epnu(0),
\end{equation}
a quadratic equation in $k$ (a direct payoff of the Taylor truncation, which converts the transcendental equation into an algebraic one). Vieta's formulas relate the two roots $k_1,k_2$:
\begin{equation}\label{eq:vieta}
k_1+k_2=-i\alpha_\nu\theta,\qquad k_1 k_2=\epnu(0)-E.
\end{equation}
Hence $\mathrm{Re}[k_1+k_2]=\aI\theta$ and $\mathrm{Im}[k_1+k_2]=-\aR\theta$. The general OBC solution is $\psi(X)=Ae^{ik_1 X}+Be^{ik_2 X}$; imposing the boundary conditions yields
\begin{equation}\label{eq:k-roots}
k_{1,2}=\pm\frac{\pi n}{\Lx}+\frac{\aI\theta}{2}-i\frac{\aR\theta}{2}\,.
\end{equation}
The two roots share the same imaginary part and the same real offset, splitting symmetrically only by $\pm\pi n/\Lx$. Substituting Eq.~\eqref{eq:k-roots} into $\psi=A(e^{ik_1 X}-e^{ik_2 X})$ gives
\begin{equation}\label{eq:chi-OBC}
\psi_n(X)\propto e^{\kappa_\nu X}\,e^{i\aI\theta X/2}\,\sin\frac{n\pi X}{\Lx}\,,
\end{equation}
where we define the inverse skin length $\kappa_\nu\equiv\aR\theta/2$. The exponential envelope is independent of the mode index $n$: all OBC eigenstates localize on the same side with the same inverse length $\kappa_\nu$, the state-level signature of the NCSE. From Eqs.~\eqref{eq:vieta} and \eqref{eq:k-roots}, using $\aI\theta/2-i\kappa_\nu=-i\alpha_\nu\theta/2$ one directly computes $k_1 k_2=-\alpha_\nu^2\theta^2/4-\pi^2 n^2/\Lx^2$, so the OBC energies are
\begin{equation}\label{eq:E-OBC}
E_n^{\mathrm{OBC}}=\epnu(0)+\frac{\pi^2 n^2}{\Lx^2}+\!\left(\kappa_\nu+\frac{i\aI\theta}{2}\right)^{\!2}.
\end{equation}
Expanding, $\mathrm{Im}[E_n^{\mathrm{OBC}}]=\mathrm{Im}[\epnu(0)]+\kappa_\nu\aI\theta$: the imaginary part of the OBC energies is likewise independent of $n$, so all OBC eigenenergies share the same imaginary part. By contrast, the imaginary part of the PBC spectrum, $\mathrm{Im}[E_\nu^{\mathrm{PBC}}(k)]=\mathrm{Im}[\epnu(0)]+\aR\theta k$, spreads linearly in $k$ [Eq.~\eqref{eq:dispersion-PBC}], the spectral signature of the NCSE~\cite{Yao2018}. The same results can also be obtained through the similarity transformation $S=e^{-\alpha_\nu\theta X/2}$, which maps Eq.~\eqref{eq:HN} onto a Hermitian particle-in-a-box problem and recovers Eqs.~\eqref{eq:chi-OBC} and \eqref{eq:E-OBC} (see Appendix~\ref{app:exact-1D}).

The above analysis can be summarized in the following theorem:
\begin{theorem}[Unified criterion for the NCSE]\label{thm:NHSE}
Under the adiabatic approximation (automatically valid as $\Lx\to\infty$) and the leading-order Taylor expansion of $\epnu(s)$:
\begin{equation}\label{eq:NHSE-criterion}
\kappa_\nu\equiv\frac{\theta}{2}\,\mathrm{Re}[\alpha_\nu]\neq 0\;\Longleftrightarrow\;\text{NCSE}.
\end{equation}
At the state level, the OBC eigenstates Eq.~\eqref{eq:chi-OBC} localize exponentially on one side of the $X$ direction with the $n$-independent inverse length $\kappa_\nu$. At the spectral level, the OBC imaginary parts Eq.~\eqref{eq:E-OBC} are likewise $n$ independent and cannot coincide with the linear $k$ dependence of the PBC imaginary parts. At $\kappa_\nu=0$, both signatures of the NCSE disappear simultaneously.
\end{theorem}
For general potential strength, $\alpha_\nu\in\mathbb{C}$ (Proposition~\ref{prop:alpha}); quantitative numerical tests of Eq.~\eqref{eq:NHSE-criterion} are presented in Secs.~\ref{sec:numerical-verification} and \ref{sec:peierls}. The ranges of validity of the two approximations are jointly controlled by the characteristic scale $s_\mathrm{gap}\sim|\epnu-\epsilon_\mu|/\|V'\|$: for a fixed mode index $n$, the adiabatic condition ($\Delta s{=}2\pi n\theta/\Lx\ll s_{\mathrm{gap}}$) and the Taylor-truncation condition ($|\theta k|\ll s_{\mathrm{gap}}$) are both satisfied automatically in the thermodynamic limit $\Lx\to\infty$. All numerical verifications in this work adopt the global adiabatic condition $j_{\mathrm{safe}}\gtrsim N_x$, which ensures that every $X$-direction mode lies within the range of validity (the definition and derivation of $j_{\mathrm{safe}}$ are given in Appendix~\ref{app:adiabatic}); this condition is a requirement of the numerical protocol, whereas Theorem~\ref{thm:NHSE} itself holds exactly for fixed $n$ in the thermodynamic limit. Both approximations presuppose the nondegeneracy of $h(s)$ (Appendix~\ref{app:validity}).

\subsection{Numerical verification}\label{sec:numerical-verification}

To test the quantitative predictions above, we construct lattice models and diagonalize them exactly. On a square lattice (lattice constant $a=1$), the kinetic term is discretized by standard central differences, and the potential is discretized order by order through the Bopp expansion Eq.~\eqref{eq:Bopp-expansion} (the complete lattice Hamiltonian and its derivation are given in Appendix~\ref{app:lattice-general}). The physical character of each Bopp order is rigorously fixed by the algebraic factor $(-1)^k i^{k+1}$: odd orders generate non-reciprocal hopping along $X$ and are the source of the NCSE, while even orders generate only symmetric corrections that preserve reciprocity. The leading non-reciprocal hopping correction, $\delta t_n^{(1)}=\theta U'(Y_n)/(2a)$, is proportional to $\theta$ and to the local potential gradient, consistent with the first-order term of Eq.~\eqref{eq:Bopp-expansion} in the continuum theory. For a polynomial potential of degree $N$ ($U^{(k)}\equiv 0$ for $k>N$), the expansion terminates exactly at $k=N$ and the lattice Hamiltonian carries no truncation error (the complete Bopp expansions and parameter constraints for each potential are given in Appendix~\ref{app:lattice-general}).

In the weak-potential limit ($\max|U|\ll\Delta\epsilon_Y\sim\pi^2/\Ly^2$, i.e., the imaginary potential is much smaller than the $Y$-direction energy scale; see Appendix~\ref{app:weak-potential} for the derivation), $\phi_\nu^R$ approaches the real eigenfunction $\phi_\nu^{(0)}$, the $Y$-direction wave functions are barely modified by the imaginary potential, and $\alpha_\nu$ in Proposition~\ref{prop:alpha} reduces to a real number:
\begin{equation}\label{eq:kappa-perturbative}
\alpha_\nu^{\mathrm{weak}}=\langle\nu|U'|\nu\rangle\in\mathbb{R},
\end{equation}
where $|\nu\rangle\equiv\phi_\nu^{(0)}$ are the box eigenstates at $V=0$. Each band $\nu$ then realizes an effective one-dimensional Hatano--Nelson model, with the inverse skin length given by Eqs.~\eqref{eq:NHSE-criterion} and \eqref{eq:kappa-perturbative}:
\begin{equation}\label{eq:kappa-weak}
\kappa_\nu^{\mathrm{weak}}=\frac{\theta}{2}\,\langle\nu|U'|\nu\rangle\in\mathbb{R}.
\end{equation}
The physical content of this formula is transparent: what determines the skin strength of the $\nu$th band is the potential gradient $U'$ weighted by the density of that transverse mode, and the noncommutativity $\theta$ converts this averaged gradient into the strength of the imaginary gauge field along $X$.

\textbf{Linear potential} $V=ig\hat{y}$: $U'=g$ is constant, $\langle\nu|U'|\nu\rangle=g$ for all $\nu$, and $\kappa_\nu=g\theta/2$, in agreement with the exact solution of Sec.~\ref{sec:linear-example}. Only the $k=1$ term of the Bopp expansion is nonzero, and the lattice Hamiltonian reduces to the standard Hatano--Nelson model with uniform non-reciprocal hopping; the numerical verification has already been shown in Sec.~\ref{sec:linear-example} [Fig.~\ref{fig:quantitative-verification}(a)].

\textbf{Cubic potential} $V=ig(Y-\Ly/2)^3$: $U'(Y)=3g(Y-\Ly/2)^2$ is symmetric about $\Ly/2$ and strictly positive, so $\kappa_\nu>0$ for all $\nu$. The Bopp expansion terminates exactly at $k=3$. Equation~\eqref{eq:kappa-weak} gives $\kappa_\nu=3g\theta\sigma_\nu^2/2$, with (see Appendix~\ref{app:cubic} for the derivation)
\begin{equation}\label{eq:sigma-nu}
\sigma_\nu^2=\frac{\Ly^2}{12}\!\left(1-\frac{6}{\nu^2\pi^2}\right).
\end{equation}
Here $\kappa_\nu$ is mode dependent: the lowest band ($\nu=1$) exhibits the weakest skin effect ($\sigma_1^2\approx 0.033\,\Ly^2$), highly excited bands approach the constant $\kappa_\infty=g\theta\Ly^2/8$ ($\sigma_\nu^2 \to \Ly^2/12$), and the ratio of skin strengths is $\kappa_\infty/\kappa_1\approx 2.55$. Figure~\ref{fig:quantitative-verification}(b) shows the scatter plot of the band-resolved $\kappa_\nu^{\mathrm{num}}$ against the prediction of Eq.~\eqref{eq:kappa-weak} in the weak-potential limit; the data points fall precisely on the $y{=}x$ reference line. On a finite lattice, $R_\nu\equiv\kappa_\nu/\kappa_1$ displays an arch-shaped structure symmetric about $\nu=(N_y+1)/2$ [Fig.~\ref{fig:quantitative-verification}(c)], which originates from Nyquist aliasing of the discrete sampling: the lattice box eigenstates obey the probability identity
\begin{equation}\label{eq:Nyquist-identity}
[\phi_\nu^{(0)}(n)]^2=[\phi_{N_y+1-\nu}^{(0)}(n)]^2,\quad\forall\,n,
\end{equation}
where $\phi_\nu^{(0)}(n)\propto\sin(\nu\pi n/(N_y+1))$, which follows directly from the trigonometric identity $\sin(\pi n-x)=(-1)^{n+1}\sin x$ for integer $n$. Inserting it into Eq.~\eqref{eq:kappa-perturbative} yields $\alpha_{N_y+1-\nu}=\alpha_\nu$, i.e., $R_{N_y+1-\nu}=R_\nu$. This identity is a lattice property of the box eigenstates themselves, independent of the specific potential, so the band-resolved $\kappa_\nu$ of the periodic imaginary potentials display the same arch symmetry (see Sec.~\ref{sec:peierls}). The continuum closed-form expressions [such as Eq.~\eqref{eq:sigma-nu}] replace discrete sums by integrals, implicitly assuming $\nu/N_y\to 0$; they therefore miss the lattice Nyquist aliasing at the high-$\nu$ end ($\nu\sim N_y$), where lattice effects are significant, while agreeing completely with the lattice result for $\nu\ll N_y$.

At general potential strength [Fig.~\ref{fig:quantitative-verification}(d)], $\alpha_\nu$ is computed from the numerical eigenstates of $h(0)$; Eq.~\eqref{eq:NHSE-criterion} tracks the band-to-band oscillatory structure of the numerical results, whereas Eq.~\eqref{eq:kappa-weak}, built on the real box eigenstates, deviates completely, confirming the validity of Theorem~\ref{thm:NHSE} beyond the weak-potential regime. Quantitative verification for periodic imaginary potentials is presented in Sec.~\ref{sec:peierls}.

\section{Exact Symmetry Criterion}\label{sec:symmetry}
Theorem~\ref{thm:NHSE} captures only the leading order $\epnu'(0)$ of the non-reciprocal function $f_\nu(s)$ and would miss the NCSE in the nongeneric situation where $\epnu'(0)=0$ but higher odd-order terms survive. The following theorem provides a criterion that covers all orders and holds exactly for finite-size systems.

\begin{theorem}[Symmetry criterion]\label{thm:symmetry}

(i)~(State level) If $V(\Ly-Z)=V(Z)$ $\forall\,Z\in\mathbb{R}$, then $[\R,H]=0$. The $X$-direction density of every OBC eigenstate satisfies $\rho(\Lx-X)=\rho(X)$: no NCSE.

(ii)~(Spectral level) If $V(\Ly-Z)=V(Z)$ $\forall\,Z\in\mathbb{R}$ and $h(s)$ is nondegenerate for all $s\in\mathbb{R}$, then $\epnu(s)=\epnu(-s)$ $\forall\,\nu,s$: the dispersion is fully reciprocal.

(iii)~If the symmetry condition fails, then for generic potentials in function space, $\exists\,\nu$ such that $\mathrm{Re}[\alpha_\nu]\neq 0$: the NCSE appears.
\end{theorem}
Here $\R$ is the two-dimensional point-reflection operator about the center of the rectangle $[0,\Lx]\times[0,\Ly]$, $(\R\Psi)(X,Y)\equiv\Psi(\Lx-X,\,\Ly-Y)$, satisfying $\R^2=\mathbf{1}$ with eigenvalues $\pm 1$. The Bopp shift couples $\hat{Y}$ and $\hat{P}_X$ into the composite coordinate $\hat{Q}\equiv\hat{Y}+\theta\hat{P}_X$; only flipping both simultaneously realizes $\R\hat{Q}\R^{-1}=\Ly-\hat{Q}$, which is why a two-dimensional point reflection, rather than a reflection along $Y$ alone, is required. The verification of the operator transformation rules and the complete proof are given in Appendix~\ref{app:symmetry-proof}. The geometric content of the symmetry condition can be read off directly from the sliding-window picture [Fig.~\ref{fig:mechanism}(c)]: the reflection $Z\mapsto\Ly-Z$ maps the window $[s,\,\Ly+s]$ onto $[-s,\,\Ly-s]$; when the potential landscape is mirror symmetric about the box center $\Ly/2$, the $\pm s$ windows are mirror images of each other and necessarily isospectral; this is the geometric essence of the similarity transformation in (ii). Conversely, mirror asymmetry renders the $\pm s$ windows spectrally inequivalent, which is the content of (iii).

Theorem~\ref{thm:symmetry} holds rigorously for finite-size systems, independent of the thermodynamic limit and of the adiabatic approximation, its key advantage over Theorem~\ref{thm:NHSE}; numerical tests are presented in Sec.~\ref{sec:symmetry-robustness}. Despite their different domains of validity, Theorem~\ref{thm:symmetry} strictly implies the no-NCSE prediction of Theorem~\ref{thm:NHSE}: if the symmetry condition holds, Theorem~\ref{thm:symmetry}(ii) gives $\epnu(s)=\epnu(-s)$, i.e., $\epnu(s)$ is an even function of $s$ whose odd-order derivatives all vanish. In particular $\epnu'(0)=i\alpha_\nu=0$, hence $\kappa_\nu=0$, recovering the prediction of Theorem~\ref{thm:NHSE}. The complete chain of implications reads
\begin{equation}\label{eq:hierarchy}
\begin{aligned}
V(\Ly\!-\!Z)=V(Z)\xRightarrow{\text{Thm~\ref{thm:symmetry}}}\epnu(s)=\epnu(-s)\Rightarrow\alpha_\nu=0
\end{aligned}
\end{equation}
None of the converses holds in general: $\kappa_\nu=0$ does not imply $\alpha_\nu=0$ ($\alpha_\nu$ may be purely imaginary); $\alpha_\nu=0$ does not imply $\epnu(s)=\epnu(-s)$ (higher odd-order terms may survive). The only situation in which the two theorems yield different predictions is an asymmetric $V$ fine-tuned such that $\epnu'(0)$ happens to vanish: Theorem~\ref{thm:NHSE} then misses the NCSE driven by higher-order terms ($\kappa\sim\theta^3$), whereas Theorem~\ref{thm:symmetry} covers all orders. Whenever Theorem~\ref{thm:NHSE} yields $\kappa_\nu\neq 0$, the contrapositive guarantees that $V$ violates the symmetry condition, so the conclusions of the two theorems agree.

Two remarks delimit the scope of Theorem~\ref{thm:symmetry}. (a)~The finite interval with OBC along $Y$ is the structural prerequisite for the theorem to have nontrivial content: in the absence of boundaries, or under PBC, $h(s)$ is unitarily equivalent to $h(0)$ and $\epnu(s)$ is independent of $s$, so the symmetry criterion degenerates into a trivial statement (Sec.~\ref{sec:perturbative}; Appendix~\ref{app:BC-necessity}). (b)~The theorem only decides the presence or absence of the NCSE; the quantitative evaluation of $\kappa_\nu$ requires Theorem~\ref{thm:NHSE}; the two theorems play complementary roles.

\begin{figure}[ptb]
\centering
{\includegraphics[width=1\linewidth]{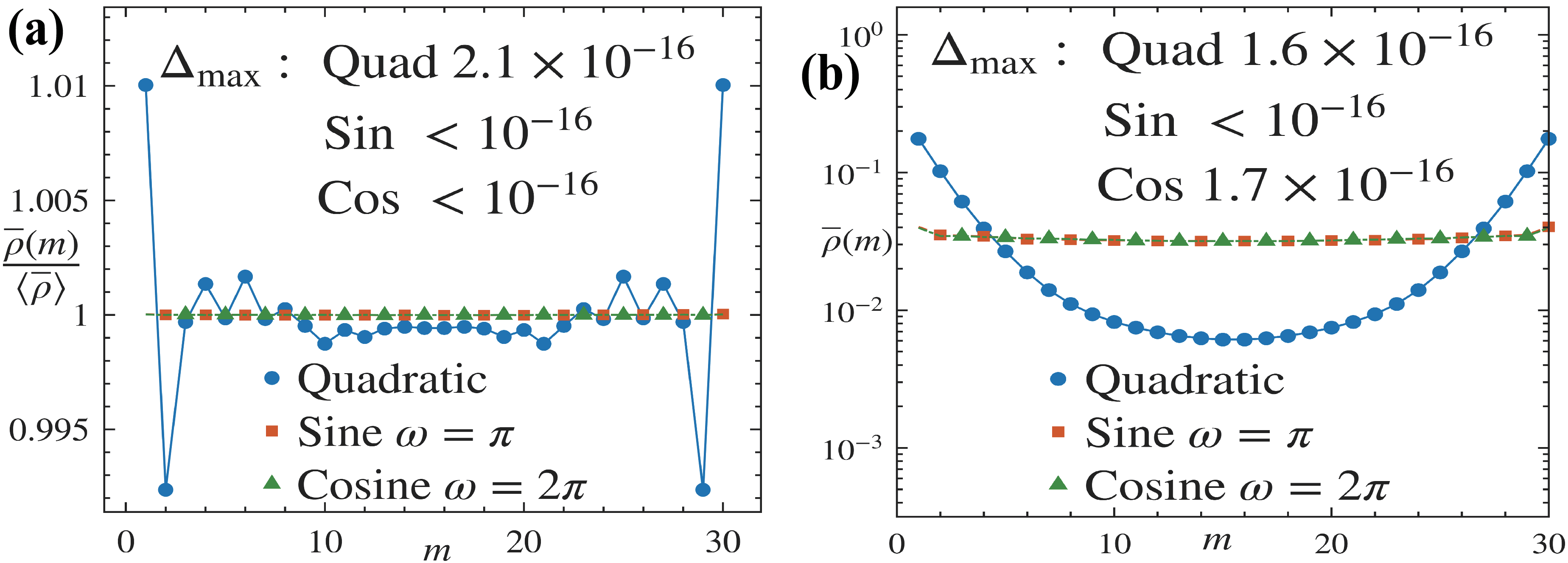}}
\caption{Numerical verification of the exact protection by Theorem~\ref{thm:symmetry}, for three symmetric potentials [blue circles: quadratic potential $ig(Y{-}\Ly/2)^2$; red squares: sine potential with $\omega{=}\pi$; green triangles: cosine potential with $\omega{=}2\pi$], $N_x{=}N_y{=}30$.
(a)~Mild parameters, $g{=}10^{-3}$ (both the weak-potential and adiabatic conditions are satisfied): the vertical axis is $\bar{\rho}(m)/\langle\bar{\rho}\rangle$ (density normalized by its spatial mean, linear scale). In the weak-potential regime the densities of all three potentials are nearly uniform, and normalizing to the mean renders the ${\sim}1\%$ shape differences discernible on a linear scale; quadratic potential, $\theta{=}0.5$; sine potential, $\theta{=}\Ly/\pi$; cosine potential, $\theta{=}\Ly/(2\pi)$.
(b)~Extreme parameters, $g{=}0.1$: the vertical axis is $\bar{\rho}(m)$ (semilogarithmic scale); quadratic potential with $\theta{=}3$, for which both the weak-potential and adiabatic conditions fail. In both parameter regimes the symmetry accuracy is $\Delta_{\max}{\lesssim}10^{-16}$, i.e., the density symmetry holds exactly for finite-size systems, independent of any approximation.}
\label{fig:symmetry-verification}
\end{figure}
\subsection{Numerical tests of the exact protection}\label{sec:symmetry-robustness}
Does the symmetry protection of Theorem~\ref{thm:symmetry} remain exact when the adiabatic approximation, the Taylor truncation, and even the weak-potential condition all fail? The answer is affirmative: the proof relies solely on the $\mathcal{R}$ symmetry and the nondegeneracy of $h(s)$, invoking no approximation whatsoever.

We test this exact protection with three structurally distinct symmetric potentials: (i)~the polynomial imaginary potential $V=ig(Y-\Ly/2)^2$, which satisfies $V(\Ly-Z)=V(Z)$, with the Bopp expansion terminating exactly at $k=2$; (ii)~the sine imaginary potential $V=ig\sin(Y/\ell)$ at $\omega\equiv\Ly/\ell=\pi$, for which $\sin((\Ly-Z)/\ell)=\sin(\pi-Z/\ell)=\sin(Z/\ell)$; (iii)~the cosine imaginary potential $V=ig\cos(Y/\ell)$ at $\omega=2\pi$, for which $\cos((\Ly-Z)/\ell)=\cos(2\pi-Z/\ell)=\cos(Z/\ell)$, likewise satisfying the symmetry condition. The three potentials are completely different in functional form (polynomial, sine, and cosine), yet the symmetry condition $V(\Ly-Z)=V(Z)$ uniformly guarantees the absence of the NCSE.

As shown in Figs.~\ref{fig:symmetry-verification}(a) and \ref{fig:symmetry-verification}(b), both for mild parameters ($g{=}10^{-3}$, with the weak-potential and adiabatic conditions satisfied) and for extreme ones [$g{=}0.1$; for the quadratic potential the weak-potential condition is violated by a factor of $\max|U|/\Delta\epsilon_Y\approx 680$ and the adiabatic condition by $\Delta s/s_\mathrm{gap}\approx 57$; for the sine and cosine potentials the weak-potential violation is about a factor of $3$], the eigenstate-averaged density $\bar\rho(m)$ along $X$ coincides with its mirror image $\bar\rho(N_x+1-m)$ on every lattice site, with the maximal asymmetry
\begin{equation}\label{eq:sym-error}
\Delta_{\max}\equiv\frac{\max_m|\bar\rho(m)-\bar\rho(N_x+1-m)|}{\max_m\bar\rho(m)}\lesssim 10^{-16},
\end{equation}
at the level of machine precision. This directly confirms the exactness of Theorem~\ref{thm:symmetry}(i): the density symmetry $\rho(\Lx-X)=\rho(X)$ is an algebraic consequence of the $\mathcal{R}$ symmetry, holding rigorously for finite-size systems regardless of the potential strength, the magnitude of $\theta$, or any perturbative approximation.

\section{Symmetry-protected measurement of the noncommutativity parameter}\label{sec:peierls}
The exact symmetry protection of Theorem~\ref{thm:symmetry} does more than decide the presence or absence of the NCSE: it predicts a hitherto unreported phenomenon: when the reflection symmetry of the imaginary potential is continuously modulated by the system geometry, the skin direction of all modes can flip collectively at specific values of a geometric parameter. The flip point is determined purely by geometry and is independent of the potential strength, furnishing a zero-crossing measurement scheme intrinsically robust against systematic errors. In this section we establish an exact lattice realization of this prediction for periodic imaginary potentials (Sec.~\ref{sec:peierls-construction}), analyze the collective flip for the sine potential (Sec.~\ref{sec:sine-potential}), and reveal the symmetry-protected nature of the flip through the contrast with the cosine potential (Sec.~\ref{sec:cosine-potential}).

\begin{figure*}[ptb]
\centering
{\includegraphics[width=1\linewidth]{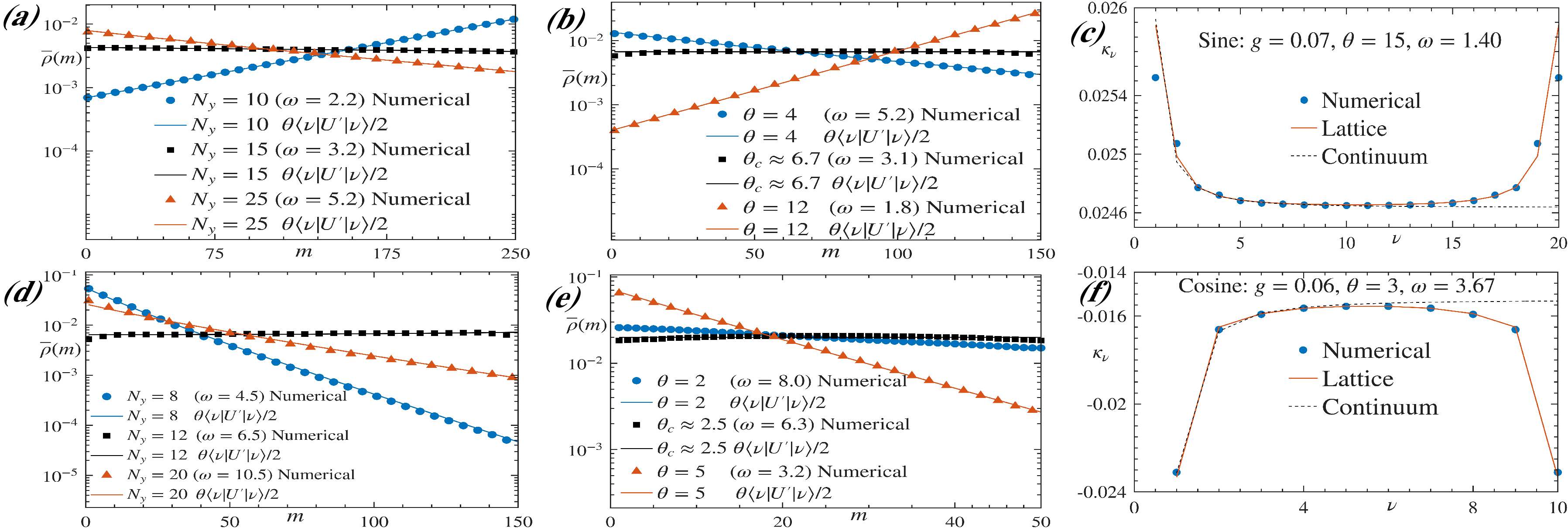}}
\caption{Skin-direction reversal for periodic imaginary potentials and quantitative verification of $\kappa_\nu$. Square lattice ($a_0{=}b_0{=}1$), exact Peierls mapping (commensurability condition $\theta/\ell{=}a_0$). In all density panels the horizontal axis is the $X$-direction site index $m$ and the vertical axis is the eigenstate-averaged density $\bar{\rho}(m)$ (semilogarithmic scale); circles/squares/triangles, numerical results; solid lines, $\bar{\rho}(m)$ predicted by Theorem~\ref{thm:NHSE}.
In (c) and (f) the horizontal axis is the band index $\nu$ and the vertical axis is the inverse skin length $\kappa_\nu$.
(a),(b)~Direction flip for the sine potential $V{=}ig\sin(\hat{y}/\ell)$: (a)~fixed $\theta{=}5$, $g{=}0.03$, $N_x{=}250$, varying $N_y\in\{10,15,25\}$ ($\omega{=}\Ly/\ell=2.2,3.2,5.2$), taking $\omega$ across the symmetry point $\pi$;
(b)~fixed $N_y{=}20$, $g{=}0.05$, $N_x{=}150$, varying $\theta\in\{4,\theta_c{\approx}6.68,12\}$, corresponding to $\omega\in\{5.25,\pi,1.75\}$.
(d),(e)~Absence of a collective flip for the cosine potential $V{=}ig\cos(\hat{y}/\ell)$: (d)~fixed $\theta{=}2$, $g{=}0.15$, $N_x{=}150$, varying $N_y\in\{8,12,20\}$ ($\omega{=}4.5$, $6.5$, $10.5$), taking $\omega$ across $2\pi$;
(e)~fixed $N_y{=}15$, $g{=}0.1$, $N_x{=}50$, varying $\theta\in\{2,\theta_c{\approx}2.55,5\}$ ($\omega{=}8$, $2\pi$, $3.2$).
(c),(f)~Band-resolved numerical and analytic $\kappa_\nu$ for the sine and cosine potentials, respectively: blue circles, $\kappa_\nu$ extracted band by band from exact diagonalization; red solid lines, lattice weak-potential prediction; black dashed lines, continuum weak-potential prediction. All bands share the same sign of $\kappa_\nu$ (collective skin localization). (c)~$N_x{=}80$, $N_y{=}20$; the sign of $\kappa_\nu$ is set by $\sin\omega$. (f)~$N_x{=}100$, $N_y{=}10$; the sign is set by $(1{-}\cos\omega){\geq}0$.
}
\label{fig:periodic-verification}
\end{figure*}

\subsection{Exact lattice model via Peierls substitution}\label{sec:peierls-construction}
The reason for choosing periodic potentials is that for non-polynomial potentials the Bopp expansion Eq.~\eqref{eq:Bopp-expansion} no longer terminates, and the lattice Hamiltonian in principle carries truncation errors. Periodic imaginary potentials, however, provide a non-polynomial example that can be treated exactly: the Baker--Campbell--Hausdorff formula for exponential operators degenerates into an exact product when $[\hat{Y},\hat{P}_X]=0$, yielding an exact lattice representation without any Taylor expansion. Take $V(\hat{y})=ig\sin(\hat{y}/\ell)$ as an example. Using Euler's formula $ig\sin(\hat{y}/\ell)=\frac{g}{2}(e^{i\hat{y}/\ell}-e^{-i\hat{y}/\ell})$ and the Bopp shift $\hat{y}=\hat{Y}+\theta\hat{P}_X$, the exponential operators factorize exactly thanks to $[\hat{Y},\hat{P}_X]=0$:
\begin{equation}\label{eq:exp-factorize}
e^{\pm i\hat{y}/\ell}=e^{\pm i\hat{Y}/\ell}\cdot e^{\pm i\theta\hat{P}_X/\ell}.
\end{equation}
In the coordinate representation, $e^{i\theta\hat{P}_X/\ell}=e^{\theta\partial_X/\ell}$ is the operator translating $X$ by $\theta/\ell$. Requiring this translation to land exactly on lattice sites imposes the commensurability condition $\theta/\ell=a_0$ ($a_0$ being the lattice constant along $X$). Equation~\eqref{eq:exp-factorize} then acts on the lattice sites $(X_m,Y_n)$ as $e^{i\hat{y}/\ell}\psi_{m,n}=e^{in\phi}\psi_{m+1,n}$ and $e^{-i\hat{y}/\ell}\psi_{m,n}=e^{-in\phi}\psi_{m-1,n}$, with the Peierls phase $\phi\equiv b_0/\ell=b_0 a_0/\theta$ ($b_0$ being the lattice constant along $Y$). The exact lattice action of the sine imaginary potential is therefore
\begin{equation}\label{eq:sine-lattice-action}
V(\hat{y})\psi_{m,n}=\tfrac{g}{2}\bigl(e^{in\phi}\psi_{m+1,n}-e^{-in\phi}\psi_{m-1,n}\bigr).
\end{equation}

Combining this exact potential term with the standard discretization of the kinetic energy, Eq.~\eqref{eq:kinetic-lattice} (allowing different lattice constants $a_0$, $b_0$ along $X$ and $Y$), yields the exact lattice Hamiltonian for the sine potential (dropping the constant on-site energy $\varepsilon_0=2t_x+2t_y$):
\begin{equation}\label{eq:sine-full-lattice}
\begin{aligned}
H_{\sin}&=\!\sum_{m,n}\bigl[J_+(n)c_{m,n}^\dagger c_{m+1,n}+J_-(n)c_{m,n}^\dagger c_{m-1,n}\\
&-t_y(c_{m,n}^\dagger c_{m,n+1}+\mathrm{h.c.})\bigr],
\end{aligned}
\end{equation}
where the $X$-direction hopping amplitudes are $J_\pm(n)\equiv-t_x\pm\tfrac{g}{2}e^{\pm in\phi}$, with $t_x=1/a_0^2$ and $t_y=1/b_0^2$ the hopping integrals along $X$ and $Y$, respectively. Equation~\eqref{eq:sine-full-lattice} shares the same structure as the generic lattice Hamiltonian Eq.~\eqref{eq:lattice-general}, with the NCSE driven entirely by the non-reciprocal hopping along $X$, $J_+(n)-J_-(n)=g\cos(n\phi)$. The difference is that Eq.~\eqref{eq:sine-full-lattice} is exact: under the commensurability condition, the exponential operators $e^{\pm i\theta\hat{P}_X/\ell}$ equal the lattice translation operators exactly, all Bopp orders are automatically resummed, and no truncation error exists.

\subsection{Sine potential: collective direction flip}\label{sec:sine-potential}

Applying the weak-potential criterion Eq.~\eqref{eq:kappa-weak} to the sine potential, with $U'=g\cos(Y/\ell)/\ell$,
\begin{equation}\label{eq:kappa-sine-derivation}
\kappa_\nu=\frac{\theta}{2\ell}\,g\,\underbrace{\frac{2}{\Ly}\int_0^{\Ly}\sin^2\!\left(\frac{\nu\pi Y}{\Ly}\right)\cos\!\left(\frac{Y}{\ell}\right)dY}_{\displaystyle\equiv\,I_\nu(\omega)},
\end{equation}
where we use the commensurability condition $\theta/\ell=a_0$. The matrix element $I_\nu(\omega)$ can be evaluated exactly via product-to-sum identities, giving $I_\nu(\omega)=4\nu^2\pi^2\sin\omega/[\omega(4\nu^2\pi^2-\omega^2)]$, and hence
\begin{equation}\label{eq:kappa-sine}
\kappa_\nu=\frac{2ga_0\nu^2\pi^2\sin\omega}{\omega(4\nu^2\pi^2-\omega^2)},\qquad\omega\equiv\frac{\Ly}{\ell}.
\end{equation}
For $\omega<2\pi$, the sign of $\kappa_\nu$ for every mode is fixed solely by $\sin\omega$. The denominator $4\nu^2\pi^2-\omega^2$ vanishes at $\omega=2\nu\pi$, but this singularity is removable [$I_\nu(2\nu\pi)=-1/2$; see Appendix~\ref{app:sine-matrix}], and $\kappa_\nu$ is everywhere finite and continuous as a function of $\omega$. The reflection-symmetry condition $V(\Ly-Z)=V(Z)$ is equivalent to $\omega=(2k+1)\pi$, at which $\sin\omega=0$ and $\kappa_\nu=0$, consistent with Theorem~\ref{thm:symmetry}. The geometric meaning of this condition is immediately visible in the sliding-window picture: the mirror axes of $\sin(Z/\ell)$ sit at its extrema $Z=(m+\tfrac{1}{2})\pi\ell$, while the box center sits at $\Ly/2=\omega\ell/2$; the two coincide if and only if $\omega=(2k+1)\pi$ [Fig.~\ref{fig:mechanism}(c)]. Tuning $\omega$ (through $\Ly$ or $\theta$) sweeps the mirror axis of the potential landscape through the box center: at each crossing, the gain imbalance between the $\pm s$ windows changes sign and the skin direction of all modes flips collectively; this is the fundamental reason why the flip point is a purely geometric quantity, independent of the potential strength $g$ [the mirror axes of the cosine potential sit at $Z=m\pi\ell$, corresponding to $\omega=2k\pi$; see Sec.~\ref{sec:cosine-potential}]. Near the first symmetry point ($\omega=\pi+\delta$, $|\delta|\ll 1$):
\begin{equation}\label{eq:direction-flip}
\kappa_\nu\approx-\frac{2ga_0\nu^2}{(4\nu^2-1)\pi}\,\delta,\qquad\delta=\frac{\Ly}{\ell}-\pi.
\end{equation}
The skin direction of every mode is fixed solely by the sign of $\delta$ and reverses at $\omega=\pi$: for $g>0$, $\delta>0$ ($\omega>\pi$) corresponds to $\kappa_\nu<0$ (localization at the left end), while $\delta<0$ ($\omega<\pi$) corresponds to $\kappa_\nu>0$ (localization at the right end). Beyond the flip itself, the parameter dependence of Eq.~\eqref{eq:kappa-sine} reveals an important difference between periodic and polynomial imaginary potentials regarding the visibility of the skin effect. For $\omega\gg 2\pi$, $|\kappa_\nu|\sim ga_0\nu^2\pi^2/\omega^3$: the cubic growth of the denominator strongly suppresses the skin effect. A large $\omega$ means the sine potential completes many full periods on $[0,\Ly]$, and the alternating-sign contributions of the potential gradient $U'(Y)\propto\cos(Y/\ell)$ largely cancel in the wave-function-weighted average $\langle\nu|U'|\nu\rangle$. Producing macroscopically observable accumulation ($|\kappa_\nu|\Lx\gtrsim 1$) requires $\omega\sim O(1)$ (i.e., $\theta\sim\Ly$), where the sine potential completes less than one period over the interval and the cancellation no longer operates. By contrast, polynomial imaginary potentials exhibit $\kappa_\nu\propto g\theta$, growing linearly without geometric suppression (e.g., $\kappa=g\theta/2$ for the linear potential and $\kappa\propto g\theta\Ly^2$ for the cubic potential), so that even $\theta\sim 1$ produces a pronounced skin effect; this difference originates from the absence of sign-alternating oscillatory structure in polynomial potential gradients on $[0,\Ly]$.

Figures~\ref{fig:periodic-verification}(a)--\ref{fig:periodic-verification}(c) present the exact-diagonalization verification of these predictions on the lattice. In (a), varying $N_y$ steers $\omega=(N_y{+}1)/\theta$ across the symmetry point $\pi$, and the density flips from the right end to the left end. In contrast, (b) demonstrates the complementary flip path by varying $\theta$, with $\theta_c{=}(N_y{+}1)/\pi$ the symmetry point. The band-resolved $\kappa_\nu$ in (c) confirms quantitatively that all bands share the same sign of $\kappa_\nu$ (collective skin localization), fixed solely by $\sin\omega$.

The exact lattice mapping above relies on the commensurability condition $\theta/\ell=a_0$. Away from commensurability ($\theta/\ell=a_0+\delta_a$), the lattice translation $e^{\theta\partial_X/\ell}$ no longer lands exactly on integer lattice sites, and Eq.~\eqref{eq:exp-factorize} requires an interpolation approximation. For small deviations $|\delta_a|\ll a_0$, the non-reciprocal hopping amplitude acquires a relative correction of $O(\delta_a/a_0)$, and the flip point shifts from $\omega=(2k+1)\pi$ to $\omega=(2k+1)\pi+O(\delta_a N_y)$. The zero-crossing scheme is therefore intrinsically robust against uncertainties in the imaginary-potential strength $g$ (the flip point does not depend on $g$) but linearly sensitive to the commensurability deviation $\delta_a$. On synthetic-dimension platforms, $\theta/\ell$ corresponds to an engineered ratio of lattice hopping parameters, whose control accuracy (typically better than $1\%$~\cite{Liang2022,Wang2021winding}) far exceeds the calibration accuracy of the gain-loss strength $g$, so the core advantage of the zero-crossing scheme survives in realistic experiments.

Among the existing direction-reversal mechanisms of the NHSE, Li \emph{et al.}~\cite{Li2022reversal} showed that coherent coupling between two non-Hermitian chains can reverse the skin direction of all modes; Yuce and Ramezani~\cite{Yuce2024bipolar} found that bipolar skin states generated by long-range asymmetric coupling reverse through a topological phase transition; and Yang and Lee~\cite{Yang2026transverseSwitch} demonstrated that transverse boundary conditions can act as a nonlocal switch controlling longitudinal skin accumulation. Although the last mechanism bears a superficial structural resemblance to the NCSE in that a transverse degree of freedom influences longitudinal skin localization, its cross-direction control still requires engineered non-Hermitian lattice parameters, with the flip point depending on the specific non-Hermitian coupling strengths. The collective flip of the NCSE, by contrast, is fixed entirely by the geometric condition $\omega=(2k+1)\pi$, independent of the imaginary-potential strength $g$, a distinction rooted in the fact that the non-reciprocity itself is a corollary of the algebra $[\hat{x},\hat{y}]=i\theta$ rather than an independent input.
\subsection{Cosine potential: absence of collective flip}\label{sec:cosine-potential}

The inverse skin length of the cosine imaginary potential $V{=}ig\cos(\hat{y}/\ell)$ is given by Eq.~\eqref{eq:kappa-cosine}; the essential difference is that the numerator factor $(1{-}\cos\omega)\geq 0$ is non-negative, touching zero quadratically at the symmetry points $\omega{=}2k\pi$ without changing sign. In Figs.~\ref{fig:periodic-verification}(d) and \ref{fig:periodic-verification}(e), $\omega$ crosses $2\pi$ yet the direction of the total density does not flip, in direct contrast with the collective flip of the sine potential in Figs.~\ref{fig:periodic-verification}(a) and \ref{fig:periodic-verification}(b) as $\omega$ crosses $\pi$. The band-resolved quantitative verification of $\kappa_\nu$ in Fig.~\ref{fig:periodic-verification}(f) agrees with Theorem~\ref{thm:NHSE}. The contrast between the sine and cosine potentials further underscores the symmetry-protected nature of the direction flip.

\section{Summary and Discussion}\label{sec:discussion}

Coordinate noncommutativity $[\hat{x},\hat{y}]=i\theta$ converts, through the Bopp shift, an imaginary potential along $\hat{y}$ into an imaginary gauge field along $X$, so that the non-reciprocity along the skin direction is uniquely fixed by the noncommutative algebra and the $\hat{y}$-direction imaginary potential rather than introduced independently, giving rise to the noncommutative skin effect (NCSE). Its physical core is a momentum-resolved translation of the potential window: the noncommutativity converts the momentum $k$ into a displacement $\theta k$ of the transverse Dirichlet box along the potential landscape, the OBC render the displaced windows unitarily inequivalent, and the mirror asymmetry of the landscape then turns the difference between the $\pm k$ windows into direction-dependent gain (Fig.~\ref{fig:mechanism}). Theorem~\ref{thm:NHSE} provides the inverse skin length $\kappa_\nu=\theta\mathrm{Re}[\alpha_\nu]/2$, exact in the thermodynamic limit; Theorem~\ref{thm:symmetry} establishes the reflection symmetry of the potential, $V(\Ly-Z)=V(Z)$, as an exact criterion for the existence of the NCSE (rigorously valid for finite systems); the two coincide at $\kappa_\nu=0$. The sine imaginary potential flips the skin direction collectively at $\omega\equiv\Ly/\ell=(2k+1)\pi$; the cosine potential, whose factor $(1-\cos\omega)\geq 0$ never changes sign, lacks this effect. The key differences between the NCSE and existing NHSE mechanisms are summarized item by item in Table~\ref{tab:NCSE-vs-HN}.

At $\theta=0$ the cross-direction transmission channel closes, and the localization behavior along $x$ disappears no matter how strong the imaginary potential along $y$ is. It is worth noting that what drives the NCSE is the symplectic structure $\theta$ of the coordinate algebra, not the metric: on a general oblique lattice, $[\hat{x}_1,\hat{x}_2]=i\theta$ produces the imaginary gauge field all the same, and orthogonal Cartesian coordinates are merely the most economical choice for the analysis, just as a Peierls flux is determined by the unit-cell area rather than by its angles.

\begin{table}[h]
\centering
\caption{Key differences between the NCSE and conventional NHSE mechanisms, including asymmetric hopping~\cite{Hatano1996}, gain/loss~\cite{yi_NonHermitianSkinModes_2020}, geometry-dependent skin effect~\cite{Wang2023GDSE}, spatially inhomogeneous imaginary potentials~\cite{Wei2025}, and direction reversal~\cite{Li2022reversal,Yuce2024bipolar,Yang2026transverseSwitch}.}
\label{tab:NCSE-vs-HN}
\footnotesize
\begin{ruledtabular}
\begin{tabular}{lll}
Feature & Conventional NHSE & NCSE (this work) \\
\hline
Asymmetry origin & independent input & $V(\hat y)$ via $\theta$ \\
Mechanism & skin-axis non-Hermiticity & $k$-sliding window \\
Minimal dimension & one & two (intrinsic) \\
Skin strength & skin-axis parameters & transverse bands \\
Flip point & parameter dependent & purely geometric \\
$\theta{=}0$ limit & unaffected & vanishes \\
\end{tabular}
\end{ruledtabular}
\end{table}

\subsection{NCSE as a probe of spatial noncommutativity}\label{sec:sensitivity}

The non-reciprocity of the NCSE is a mathematical necessity of any faithful representation of the algebra~\eqref{eq:commutation} (Stone--von Neumann uniqueness; Theorem~\ref{thm:no-go}), which distinguishes it, at the algebraic level, from the noncommutativity of internal degrees of freedom (such as non-Abelian gauge structures generating non-reciprocity through the noncommutativity of matrices). In recent years, exceptional-point-enhanced response~\cite{Wiersig2020review} and non-Hermitian topological sensors~\cite{Budich2020,Koch2022} have become active research directions, but in all of them the non-reciprocity is introduced by engineered parameters; the relation $\kappa\propto\theta$ of the NCSE means, instead, that the skin strength directly encodes the algebra of the underlying space.

Compared with conventional $\theta$ detection based on perturbative spectral corrections~\cite{Gamboa2001}, the advantage of the NCSE lies in a qualitative transition rather than a quantitative correction. Taking weak-potential parameters as an example ($g=3\times10^{-3}$, $\theta=0.4$, $N_x=500$): the correction to the real part of the spectrum is $\sim 10^{-5}$, whereas $|\kappa_1|N_x\approx 1$ already renders the NCSE macroscopically visible. By Theorem~\ref{thm:symmetry}, as a binary detector there is no in-principle lower bound on the detectable $\theta$ in the thermodynamic limit; the visibility threshold for a finite system is $|\kappa|N_x\gtrsim 1$. The zero-crossing scheme based on the sine potential (with the flip point $\theta_c$ independent of $g$) is intrinsically robust against systematic errors. Similar cross-coupling structures appear in more general deformed algebras~\cite{Lukierski1991}; extending the NCSE mechanism to $\kappa$-Minkowski spacetime is a direction worth exploring.

\subsection{Physical realization}
Experimental verification of the NCSE can be understood at two levels. \emph{At the first level} (effective lattice models): the Peierls-phase structure of the lattice Hamiltonian Eq.~\eqref{eq:sine-full-lattice} can be engineered on platforms such as photonic ring resonators~\cite{Wang2021winding} and topolectrical circuits~\cite{helbig_GeneralizedBulkBoundary_2020} without realizing coordinate noncommutativity, allowing direct tests of the direction-flip and symmetry-protection predictions. \emph{At the second level} (as a probe of noncommutative space): this requires the coexistence of coordinate noncommutativity (R1), second-order kinetic energy (R2), and an imaginary potential (R3). The most direct candidates (Landau-level guiding centers~\cite{Ezawa2013} and a phase-space reinterpretation) each face structural obstructions (Appendix~\ref{app:realization}).

Synthetic-dimension schemes~\cite{Ozawa2019synthetic} are free of these obstructions: cold-atom~\cite{Liang2022,Zhao2025}, photonic~\cite{Wang2021winding,Zheng2024}, and circuit~\cite{Zhang2024CLM} platforms already possess the ability to control hopping parameters and gain/loss independently. The central open problem is the synthetic-dimension realization of coordinate noncommutativity: the closest existing scheme is the non-Abelian gauge field in a synthetic frequency dimension~\cite{Cheng2025} [SU(2)], which has not yet directly realized the coordinate algebra $[\hat{x},\hat{y}]=i\theta$; introducing noncommutativity through a synthetic magnetic flux, on the other hand, would bring back the difficulties associated with Landau quantization~\cite{Lu2021} (Appendix~\ref{app:LLL}). Realizing coordinate noncommutativity without introducing an effective magnetic field, while preserving quadratic dispersion and a controllable imaginary potential, is a technical challenge that calls for new schemes.

\appendix

\section{Bopp transformation and star product}\label{app:Bopp-star}
The symmetric Bopp shift reads
\begin{equation}\label{eq:sym-Bopp}
\begin{split}
\hat{x}=\hat{X}'-\frac{\theta}{2}\hat{P}_{Y'},&\quad \hat{y}=\hat{Y}'+\frac{\theta}{2}\hat{P}_{X'},\\
\hat{p}_x=\hat{P}_{X'},&\quad \hat{p}_y=\hat{P}_{Y'}.
\end{split}
\end{equation}
The asymmetric Bopp shift is Eq.~\eqref{eq:Bopp} of the main text: $\hat{x}=\hat{X}$, $\hat{y}=\hat{Y}+\theta\hat{P}_X$, $\hat{p}_x=\hat{P}_X$, $\hat{p}_y=\hat{P}_Y$. The key difference between the two forms is the following: in the symmetric form, $\hat{x}=\hat{X}'-(\theta/2)\hat{P}_{Y'}$ mixes coordinate and momentum, the Bloch wave vector $k_x$ is no longer the physical quasimomentum, and the potential is shifted by $\theta k_x/2$; in the asymmetric form, $\hat{x}=\hat{X}$ corresponds directly to the physical coordinate, $k_x$ is the physical quasimomentum, and the potential shift is $\theta k_x$. We adopt the asymmetric form so that the wave-function profile along $X$ directly reflects the physical distribution along $\hat{x}$, rendering the analysis of the skin effect most transparent.

An equivalent formulation of noncommutative space is the Moyal--Weyl star product~\cite{Szabo2003}: the coordinates are kept as commuting real numbers $(x,y)$, but ordinary multiplication is replaced by the star product. For $[\hat{x},\hat{y}]=i\theta$, the star product is defined as
\begin{equation}\label{eq:star-def}
(f\star g)(x,y)=f\,\exp\!\left[\frac{i\theta}{2}\!\left(\overleftarrow{\partial}_x\overrightarrow{\partial}_y-\overleftarrow{\partial}_y\overrightarrow{\partial}_x\right)\right]g\,,
\end{equation}
where $\overleftarrow{\partial}$ acts to the left on $f$ and $\overrightarrow{\partial}$ acts to the right on $g$. Direct verification gives $x\star y-y\star x=i\theta$, reproducing the noncommutative algebra.

Let $V=V(y)$ depend on $y$ only. Since $\partial_x^{n-k}V=0$ for $n-k\geq 1$, only the $k=n$ terms survive in the series expansion of the star product:
\begin{equation}\label{eq:star-Bopp-equiv}
(V\star\psi)(x,y)=\sum_{n=0}^{\infty}\frac{1}{n!}\!\left(-\frac{i\theta}{2}\right)^{\!n}V^{(n)}(y)\,\partial_x^n\psi.
\end{equation}
This is precisely the Taylor expansion of $V(y-\frac{i\theta}{2}\partial_x)\psi$ under the symmetric Bopp shift. The asymmetric Bopp expansion carries coefficients $(-i\theta)^n$, differing from Eq.~\eqref{eq:star-Bopp-equiv} by a factor of $2^{-n}$, which corresponds to a different parametrization of the coordinates; the final physical results are equivalent.

\section{Representation-theoretic uniqueness}\label{app:dim-obstruction}
In the main text, the Bopp shift [Eq.~\eqref{eq:Bopp}] maps the noncommutative algebra onto two pairs of standard canonical variables $(\hat{X},\hat{P}_X)$, $(\hat{Y},\hat{P}_Y)$, with Hilbert space $L^2(\mathbb{R}^2)$. This appendix settles two fundamental questions. (a)~\emph{Dimensional obstruction}: why does a single canonical pair [i.e., $L^2(\mathbb{R})$] not suffice? Theorem~\ref{thm:no-go} below provides a rigorous proof. (b)~\emph{Representation uniqueness}: does the specific parametrization of the Bopp shift affect the physical results? The answer is no, as guaranteed by the Stone--von Neumann theorem~\cite{Gouba2009,Scholtz2009}.

\begin{theorem}[Dimensional obstruction]\label{thm:no-go}
The algebra $[\hat{x},\hat{y}]=i\theta$, $[\hat{x},\hat{p}_x]=[\hat{y},\hat{p}_y]=i$, $[\hat{p}_x,\hat{p}_y]=[\hat{x},\hat{p}_y]=[\hat{y},\hat{p}_x]=0$ admits no faithful representation on a single $L^2(\mathbb{R})$ (or any finite-dimensional truncation thereof).
\end{theorem}
\begin{proof}
Suppose a faithful representation existed on $L^2(\mathbb{R}_x)$, with $\hat{x}=x$ and $\hat{y}=-i\theta\partial_x$ (the only choice satisfying $[\hat{x},\hat{y}]=i\theta$). The condition $[\hat{y},\hat{p}_x]=0$ forces $\hat{p}_x$ to commute with $\partial_x$, hence $\hat{p}_x=f(-i\partial_x)$. Then $[\hat{x},\hat{p}_x]=i$ gives $f'=1$, so $\hat{p}_x=\hat{y}/\theta+c_1$. Similarly, $[\hat{x},\hat{p}_y]=0$ together with $[\hat{y},\hat{p}_y]=i$ yields $\hat{p}_y=-\hat{x}/\theta+c_2$. However, $[\hat{p}_x,\hat{p}_y]=-[\hat{y},\hat{x}]/\theta^2=i/\theta\neq 0$, contradicting the required $[\hat{p}_x,\hat{p}_y]=0$.
\end{proof}
The root of this obstruction is that the algebra Eq.~\eqref{eq:commutation} has four independent generators $(\hat{x},\hat{y},\hat{p}_x,\hat{p}_y)$, whose commutation relations define a step-two nilpotent Lie group with a three-dimensional center (corresponding to the three central elements $\theta$, $\hbar$, and $[\hat{p}_x,\hat{p}_y]=0$)~\cite{Gouba2009,Scholtz2009}. A single canonical pair provides only two generators, insufficient for a faithful representation of four independent operators. The minimal faithful representation therefore requires two canonical pairs $(\hat{X},\hat{P}_X)$ and $(\hat{Y},\hat{P}_Y)$, with Hilbert space $L^2(\mathbb{R}^2)$.

\textbf{Limitation of the fuzzy-plane construction}. The noncommutative relation $[\hat{x},\hat{y}]=i\theta$ is isomorphic to the harmonic-oscillator algebra $[\hat{a},\hat{a}^\dagger]=1$, and truncating the Fock space $\{|n\rangle\}_{n=0}^{N-1}$ yields finite-dimensional matrix representations of $\hat{x}$ and $\hat{y}$, namely the fuzzy plane~\cite{Madore1999} (this construction generalizes the fuzzy sphere~\cite{Madore1992} in the flat limit). However, this represents only the two operators $\hat{x}$ and $\hat{y}$. The Hamiltonian Eq.~\eqref{eq:H-original} of this work also contains the momentum operators $\hat{p}_x$, $\hat{p}_y$, which, by Theorem~\ref{thm:no-go}, cannot be consistently defined on the same Fock space: the algebra forces $\hat{p}_x$ and $\hat{p}_y$ to become functions of $\hat{x}$ and $\hat{y}$, thereby violating $[\hat{p}_x,\hat{p}_y]=0$. To represent all four operators simultaneously, one must enlarge the space to the tensor product $\mathcal{H}=\text{span}\{|n_1\rangle\otimes|n_2\rangle\}$~\cite{Scholtz2009}, at which point the Bopp shift inevitably appears.

\textbf{Uniqueness of the Bopp representation}. The Bopp shift [Eq.~\eqref{eq:Bopp}] provides an explicit isomorphism from the algebra [Eq.~\eqref{eq:commutation}] to the Weyl algebra of two standard canonical pairs $(\hat{X},\hat{P}_X)$, $(\hat{Y},\hat{P}_Y)$. Let $\varphi$ be any faithful irreducible representation of Eq.~\eqref{eq:commutation}; through this isomorphism, $\varphi$ induces an irreducible representation of the standard Weyl algebra, which by the Stone--von Neumann theorem~\cite{Gouba2009,Scholtz2009} is unitarily equivalent to the Schr\"odinger representation on $L^2(\mathbb{R}^2)$~\cite{Gouba2009}. Scholtz \emph{et al.}~\cite{Scholtz2009} reached the same conclusion from the perspective of the Hilbert--Schmidt operator space (i.e., the set of operators $A$ on $\mathcal{H}_c$ with $\mathrm{Tr}(A^\dagger A)<\infty$, where $\mathcal{H}_c$ is the Fock-space carrier of $[\hat{x},\hat{y}]=i\theta$): this space is naturally isomorphic to $\mathcal{H}_c\otimes\mathcal{H}_c^*\cong L^2(\mathbb{R}^2)$. Corollary: all faithful irreducible representations of Eq.~\eqref{eq:commutation} are unitarily equivalent; whether one adopts the fuzzy tensor-product basis~\cite{Madore1999}, a spatial-grid discretization, or star-product quantization~\cite{Szabo2003}, the resulting spectra are identical. The specific parametrization of the Bopp shift (symmetric or asymmetric) is a coordinate choice within the unitary equivalence class and does not affect physical results.

\section{Proofs and Derivations}\label{app:proofs}

\subsection{Non-Hermitian Hellmann-Feynman theorem and $\alpha_\nu$}\label{app:alpha-derivation}

This section derives Proposition~\ref{prop:alpha}. The central tool is the non-Hermitian Hellmann--Feynman theorem (NH-HFT)~\cite{Moiseyev2011,Hajong2023,Hajong2025}. For a non-Hermitian operator $H(\lambda)$ depending on a parameter $\lambda$, let its right and left eigenstates be $|\phi_\nu^R(\lambda)\rangle$ and $\langle\phi_\nu^L(\lambda)|$, satisfying $H|\phi_\nu^R\rangle=\epsilon_\nu|\phi_\nu^R\rangle$ and $\langle\phi_\nu^L|H=\epsilon_\nu\langle\phi_\nu^L|$. Under the nondegeneracy condition, the NH-HFT states that
\begin{equation}\label{eq:NH-HFT}
\frac{d\epsilon_\nu}{d\lambda}=\frac{\langle\phi_\nu^L|\frac{\partial H}{\partial\lambda}|\phi_\nu^R\rangle}{\langle\phi_\nu^L|\phi_\nu^R\rangle},
\end{equation}
where the denominator $\langle\phi_\nu^L|\phi_\nu^R\rangle$ is the biorthogonal normalization factor. The proof parallels the Hermitian case: differentiate $H|\phi_\nu^R\rangle=\epsilon_\nu|\phi_\nu^R\rangle$ with respect to $\lambda$, multiply from the left by $\langle\phi_\nu^L|$, and use $\langle\phi_\nu^L|H=\epsilon_\nu\langle\phi_\nu^L|$ to eliminate the $\partial_\lambda|\phi_\nu^R\rangle$ term. We apply the NH-HFT to $h(s)=-\frac{d^2}{dY^2}+V(Y+s)$, with the parameter $\lambda=s$. According to Eq.~\eqref{eq:h-def}, $\partial_s h|_{s=0}=V'(Y)=iU'(Y)$.

Let $h(0)=T+iU(Y)$ with $T=-\frac{d^2}{dY^2}$. Taking the complex conjugate of $h(0)\phi_\nu^R=\epnu\phi_\nu^R$ and using $T^*=T$, one finds $[T-iU](\phi_\nu^R)^*=\epnu^*(\phi_\nu^R)^*$. Since $h(0)^\dagger=T-iU$, the left eigenstate is $\phi_\nu^L\propto(\phi_\nu^R)^*$. Choosing the overall phase of $\phi_\nu^R$ such that $\int[\phi_\nu^R]^2\,dY>0$ (the nondegeneracy condition guarantees that this integral is nonzero~\cite{Kato1966}), the biorthogonal normalization $\langle\phi_\nu^L|\phi_\nu^R\rangle=1$ fixes the normalization constant
\begin{equation}\label{eq:c-nu}
c_\nu\equiv\frac{1}{\int_0^{\Ly}[\phi_\nu^R(Y)]^2\,dY}\in\mathbb{R}_{>0},
\end{equation}
i.e., $\phi_\nu^L(Y)=c_\nu\,[\phi_\nu^R(Y)]^*$. Substituting these results into the NH-HFT~\eqref{eq:NH-HFT}:
\begin{equation}
\begin{aligned}
\epnu'(0)=\langle\phi_\nu^L|iU'|\phi_\nu^R\rangle&=c_\nu\int_0^{\Ly}[\phi_\nu^R(Y)]^2\,iU'(Y)\,dY\\
&=i\,\frac{\int[\phi_\nu^R]^2 U'\,dY}{\int[\phi_\nu^R]^2\,dY}=i\alpha_\nu.
\end{aligned}
\end{equation}

\subsection{Weak-potential condition}\label{app:weak-potential}

The energy levels of the free particle in a box along $Y$ ($V=0$) are $\epsilon_\nu^{(0)}=(\nu\pi/\Ly)^2$ (taking $\hbar=1$, $m=1/2$), with the lowest level spacing $\Delta\epsilon_Y\equiv\epsilon_2^{(0)}-\epsilon_1^{(0)}=3\pi^2/\Ly^2$. The weak-potential condition requires the perturbation of the $Y$-direction wave functions by the imaginary potential to be sufficiently small that $\phi_\nu^R\approx\phi_\nu^{(0)}$ (a real function). Nondegenerate perturbation theory gives the first-order correction
\begin{equation}
\phi_\nu^{(1)}=\sum_{\mu\neq\nu}\frac{\langle\mu|V|\nu\rangle}{\epsilon_\nu^{(0)}-\epsilon_\mu^{(0)}}\phi_\mu^{(0)},
\end{equation}
whose norm satisfies $\|\phi_\nu^{(1)}\|\lesssim\max|U|/\Delta\epsilon_Y$. Convergence of the perturbation series requires
\begin{equation}\label{eq:weak-condition}
\max|U|\ll\Delta\epsilon_Y\sim\pi^2/\Ly^2\,.
\end{equation}
For the specific potentials, this condition translates into upper bounds on the coupling constant $g$ [taking $\Ly=(N_y+1)b_0$]: linear potential, $g\ll 3\pi^2/\Ly^3$; quadratic potential, $g\ll 12\pi^2/\Ly^4$; cubic potential, $g\ll 24\pi^2/\Ly^5$; sine and cosine potentials, $g\ll 3\pi^2/\Ly^2$. For polynomial potentials the constraint tightens as a power of $\Ly$ ($\max|U|$ is attained at the interval boundary), whereas for the sine and cosine potentials $|U|\leq g$ is globally bounded, and the constraint is the most lenient.

\subsection{Adiabatic condition: self-consistency}\label{app:adiabatic}

The adiabatic approximation requires the variation of $\phi_\nu(Y;s)$ over the parameter spacing $\Delta s$ between different roots $\beta_j$ to be negligible. We establish a quantitative condition by estimating the sensitivity of $\phi_\nu$ to the parameter $s$.

Viewing $h(s)=-d^2/dY^2+V(Y+s)$ as an operator family depending on the parameter $s$, first-order non-Hermitian perturbation theory~\cite{Moiseyev2011} gives the parametric derivative of the transverse wave function:
\begin{equation}\label{eq:phi-variation}
\frac{\partial\phi_\nu^R}{\partial s}=\sum_{\mu\neq\nu}\frac{\langle\phi_\mu^L|V'(Y)|\phi_\nu^R\rangle}{\epnu-\epsilon_\mu}\,\phi_\mu^R,
\end{equation}
where $\partial_s h(s)|_{s=0}=V'(Y)$. The matrix elements in the numerator satisfy $|\langle\phi_\mu^L|V'|\phi_\nu^R\rangle|\leq\|V'\|$, and the dominant contribution comes from the nearest band $\mu$, so that
\begin{equation}
\left\|\frac{\partial\phi_\nu^R}{\partial s}\right\|\sim\frac{\|V'\|}{|\epnu-\epsilon_\mu|}\,.
\end{equation}
The parameter displacement required for an $O(1)$ change of $\phi_\nu$ is therefore
\begin{equation}
s_{\mathrm{gap}}\sim\frac{|\epnu-\epsilon_\mu|}{\|V'\|}\,,
\end{equation}
i.e., the ratio of the $Y$-direction band gap to the magnitude of the potential gradient. The larger the gap and the weaker the potential gradient, the larger $s_\text{gap}$ and the more easily the adiabatic condition is satisfied. The OBC quantization gives the spacing of adjacent roots $|\beta_+-\beta_-|=2\pi n/\Lx$, corresponding to the parameter spacing $\Delta s=2\pi n\theta/\Lx$. The adiabatic condition is $\Delta s\ll s_{\mathrm{gap}}$. For a fixed mode index $n$, $\Delta s\propto 1/\Lx$ vanishes in the thermodynamic limit $\Lx\to\infty$, and the condition is satisfied automatically.

For the specific potentials, using the free-particle gap $\Delta\epsilon_Y=3\pi^2/\Ly^2$ as an order-of-magnitude estimate for $|\epnu-\epsilon_\mu|$ (this estimate is accurate when the potential is much smaller than the gap, and yields a conservative parameter range for finite potentials), the adiabatic scales read: linear potential, $s_\mathrm{gap}\sim 3\pi^2/(g\Ly^2)$ (the potential gradient $\|V'\|=g$ is uniform; the adiabatic condition is the easiest to satisfy); cubic potential, $s_\mathrm{gap}\sim 4\pi^2/(g\Ly^4)$ (the potential gradient $\|V'\|\propto\Ly^2$; the constraint is more stringent); sine and cosine potentials, $s_\mathrm{gap}\sim 3\pi^2\ell/(g\Ly^2)$ (the potential gradient $\|V'\|=g/\ell$ is globally bounded; $s_\mathrm{gap}$ depends on $\ell$). The quadratic potential satisfies $V(\Ly-Z)=V(Z)$ and has $\kappa_\nu\equiv 0$, requiring no adiabatic condition.

For the $j$th $X$-direction mode, the adiabatic condition reads $2\pi j\theta/\Lx\ll s_{\mathrm{gap}}$, i.e., the mode index must satisfy $j\ll s_{\mathrm{gap}}\Lx/(2\pi\theta)$. The band-resolved numerical extraction of $\kappa_\nu$ requires this condition to hold for the vast majority of $j$ (not merely $j=1$): projecting the two-dimensional eigenstates onto the transverse basis $\phi_\nu$ must assign every $X$-direction eigenstate to the correct band, and if the adiabatic condition fails for high-$j$ states, the projection mixes different bands and introduces systematic bias into the extracted $\kappa_\nu$. Defining the largest mode index satisfying the adiabatic condition, $j_{\mathrm{safe}}\equiv\lfloor s_{\mathrm{gap}}\Lx/(2\pi\theta)\rfloor$ ($\lfloor\cdot\rfloor$ denotes the floor function), we require $j_{\mathrm{safe}}\gtrsim N_x$ (where $\gtrsim$ means greater than or of comparable order), i.e., all $N_x$ modes fall within the range of validity of the adiabatic condition. We emphasize that this global adiabatic condition is a requirement of the numerical verification protocol, not of the theoretical validity of Theorem~\ref{thm:NHSE}, which holds exactly in the thermodynamic limit for states of fixed mode index $n$. For the cubic potential, this global adiabatic condition combined with the visibility condition $|\kappa_1|N_x\gtrsim 1$ yields the system-size scaling law $N_x\gtrsim 10(N_y-1)^2$.

\subsection{Boundary condition necessity}\label{app:BC-necessity}
Open boundaries along $Y$ are a necessary condition for the dispersion non-reciprocity $\epnu(s)\neq\epnu(-s)$. In this section we discuss in detail two issues: the classification of boundary conditions with respect to the translation operator, and the boundary conditions on the commutative coordinates.

Define the $Y$-direction translation operator $(T_s\phi)(Y)\equiv\phi(Y+s)$. For any $f(Y)$, using $(T_{-s}f)(Y)=f(Y-s)$ and the chain rule $[f(Y-s)]''=f''(Y-s)$, one directly verifies
\begin{equation}
(T_s h(0) T_{-s} f)(Y)=-f''(Y)+V(Y+s)f(Y),
\end{equation}
i.e., $T_s\,h(0)\,T_{-s}=h(s)$ holds for arbitrary potentials $V$ and $s\in\mathbb{R}$. Whether $h(s)$ and $h(0)$ are isospectral depends on the unitarity of $T_s$ on the corresponding function space.

On the unbounded line $L^2(\mathbb{R})$ and under \emph{periodic boundary conditions} $\phi(0)=\phi(\Ly)$, $T_s$ is unitary: in the former case, the change of variables gives $\langle T_sf,T_sg\rangle=\langle f,g\rangle$ directly; in the latter, $T_s$ preserves periodicity, and the integrand integrates to the same value over any interval of length $\Ly$. In both cases $h(s)$ is unitarily equivalent to $h(0)$, $\epnu(s)=\epnu(0)$, and non-reciprocity cannot arise. Under \emph{OBC}, $\phi(0)=\phi(\Ly)=0$, one has $(T_s\phi)(0)=\phi(s)\neq 0$ (generically): $T_s$ maps functions out of the space satisfying this boundary condition, the unitary equivalence is broken, and $\epnu(s)\neq\epnu(-s)$ becomes possible. The properties of $T_s$ under the various boundary conditions are classified in Table~\ref{tab:BC-classification}; the sheared contour bands in Fig.~\ref{fig:mechanism}(a) of the main text are precisely the visualization of this translation structure of $h(s)$ in the $(k,Y)$ plane.

\begin{table}[h]
\centering
\caption{Classification of the properties of the translation operator $T_s$ under different boundary conditions (BCs) along $Y$. ``Preserves BC'' indicates whether $T_s$ maps functions satisfying the BC to functions still satisfying it; ``Unitary'' indicates whether $T_s$ is unitary on the corresponding function space; ``Non-recip.'' indicates whether the dispersion non-reciprocity $\epnu(s)\neq\epnu(-s)$ can arise. Only OBC break the unitary equivalence and make non-reciprocity possible.}
\label{tab:BC-classification}
\footnotesize
\begin{ruledtabular}
\begin{tabular}{lccc}
Boundary condition & Preserves BC & Unitary & Non-recip. \\
\hline
Real line $Y\in\mathbb{R}$ & \checkmark & \checkmark & $\times$ \\
PBC $\phi(0){=}\phi(\Ly)$ & \checkmark & \checkmark & $\times$ \\
Quasiperiodic $\phi(\Ly){=}e^{i\varphi}\phi(0)$ & \checkmark & \checkmark & $\times$ \\
OBC $\phi(0){=}\phi(\Ly){=}0$ & $\times$ & $\times$ & \checkmark \\
\end{tabular}
\end{ruledtabular}
\end{table}

The OBC are imposed on the commutative coordinate $\hat{Y}$, not on the physical coordinate $\hat{y}=\hat{Y}+\theta\hat{P}_X$. This is not a choice but a necessity dictated by the algebraic structure. In the coordinate representation, $\hat{y}=Y-i\theta\partial_X$ contains a differential operator along $X$ and does not correspond to a definite spatial location, so ``vanishing of the wave function at $\hat{y}=0$'' would be the partial-differential constraint $Y\Psi=i\theta\partial_X\Psi$ rather than a geometric boundary condition. By the dimensional-obstruction theorem (Appendix~\ref{app:dim-obstruction}), all faithful irreducible representations of the algebra~\eqref{eq:commutation} are unitarily equivalent to the Bopp representation on $L^2(\mathbb{R}^2)$~\cite{Gouba2009,Scholtz2009}; the wave function $\Psi(X,Y)$ is parametrized by the commutative coordinates $(X,Y)$, and boundary conditions can only be imposed on these coordinates. In synthetic-dimension implementations, the lattice index $n$ associated with $\hat{Y}$ (such as photonic frequency modes or cold-atom momentum ladders) is the natural discrete degree of freedom, and the boundary is implemented precisely as a truncation of that synthetic-dimension index; ``OBC on $\hat{Y}$'' thus corresponds exactly to the natural boundary in experiments, rather than being a mere mathematical convenience.

The global definition of the potential ($V$ must be specified for all $Z\in\mathbb{R}$; see Sec.~\ref{sec:model}) manifests itself even more thoroughly at the lattice level: the lattice model contains only the sites and hoppings within the sample (there is no ``region outside the sample''), and the values of $V$ outside the window are encoded in the $n$-dependent structure of the $X$-direction hoppings. Taking the exact Peierls mapping of the sine potential, Eq.~\eqref{eq:sine-full-lattice}, as an example: for an $X$-direction PBC plane wave $e^{ikm}$, the hopping terms give $J_+(n)e^{ik}+J_-(n)e^{-ik}=-2t_x\cos k+ig\sin(n\phi+k)$; using the commensurability condition $\theta=a_0\ell$ and $\phi=b_0/\ell$, one finds $n\phi+k=(Y_n+s)/\ell$ (with $s=\ell k=\theta k/a_0$), i.e., the hopping structure inside the box automatically reconstructs, in momentum space, the translated potential $U(Y_n+s)$, including its values ``outside the window''. Furthermore, in the continuum theory with OBC along $X$, the roots $\beta_j$ of the non-Bloch superposition Eq.~\eqref{eq:Psi-general} are generally complex, and $\phi_\nu(Y;\theta\beta_j)$ involves complex shifts, strictly requiring $V$ to admit an analytic continuation in a complex neighborhood; the polynomial, sine, and cosine potentials studied in this work are all entire functions, so this requirement is met automatically, while for general potentials the lattice model can be taken as the fundamental definition.

\subsection{Proof of Theorem~\ref{thm:symmetry}}\label{app:symmetry-proof}

This section provides the complete proof of Theorem~\ref{thm:symmetry}.

\textbf{(i) $V(\Ly-Z)=V(Z)$ $\Rightarrow$ $[\R,H]=0$ $\Rightarrow$ no NCSE.}

We first verify the transformation rules of $\R$ on the basic operators. For the momentum operators, taking $\hat{P}_X$ as an example:
\begin{align}
(\R\hat{P}_X\R\psi)(X,Y) &= \R\!\left[-i\partial_X\psi(\Lx-X,\Ly-Y)\right] \notag\\
&= +i\frac{\partial\psi}{\partial u}\Big|_{u=X} = -\hat{P}_X\psi(X,Y),
\end{align}
and similarly $\R\hat{P}_Y\R^{-1}=-\hat{P}_Y$. For the coordinate operators, $\R\hat{X}\R^{-1}=\Lx-\hat{X}$ and $\R\hat{Y}\R^{-1}=\Ly-\hat{Y}$ follow directly from the reflection of the domain by $\R$. For the composite coordinate: $\R\hat{Q}\R^{-1}=(\Ly-\hat{Y})+\theta(-\hat{P}_X)=\Ly-\hat{Q}$.

Since $\hat{P}_X^2$ and $\hat{P}_Y^2$ are invariant under $\R$ (the squares cancel the minus signs), the kinetic term commutes with $\R$. For the potential term: $\R V(\hat{Q})\R^{-1}=V(\R\hat{Q}\R^{-1})=V(\Ly-\hat{Q})$. The symmetry condition $V(\Ly-Z)=V(Z)$ $\forall Z$ guarantees $V(\Ly-\hat{Q})=V(\hat{Q})$, hence $[\R,H]=0$.

Let $H\Psi=E\Psi$. If $E$ is nondegenerate, then $\R\Psi=\lambda\Psi$ with $\lambda=\pm 1$, and therefore $|\Psi(\Lx-X,\Ly-Y)|^2=|\Psi(X,Y)|^2$. Summing over $Y$ yields $\rho(\Lx-X)=\rho(X)$, i.e., the density is centrosymmetric along $X$, which rules out one-sided skin localization. If $E$ is degenerate, $\R$ satisfies $\R^2=\mathbf{1}$ on the degenerate subspace and can be diagonalized with eigenvalues $\pm 1$; each basis vector then satisfies the density symmetry independently.

\textbf{(ii) $V(\Ly-Z)=V(Z)$ $\Rightarrow$ $\epnu(s)=\epnu(-s)$.}

Let $\phi_\nu(Y;s)$ satisfy $h(s)\phi_\nu=\epnu(s)\phi_\nu$ [OBC: $\phi_\nu(0;s)=\phi_\nu(\Ly;s)=0$]. To connect the spectra of $h(s)$ and $h(-s)$, construct the function $\tilde\phi_\nu(Y)\equiv\phi_\nu(\Ly-Y;s)$. The boundary conditions of $\phi_\nu$ directly give $\tilde\phi_\nu(0)=\phi_\nu(\Ly;s)=0$ and $\tilde\phi_\nu(\Ly)=\phi_\nu(0;s)=0$, i.e., $\tilde\phi_\nu$ satisfies the same OBC.

Perform the change of variables $Y\to\Ly-Y$ in $h(s)\phi_\nu=\epnu(s)\phi_\nu$. In the kinetic term, the chain rule produces two minus signs that cancel ($\frac{d^2}{dY^2}[\phi(\Ly-Y)]=\phi''(\Ly-Y)$), so the kinetic term is unchanged. The potential term becomes $V(\Ly-Y+s)$. Using the symmetry condition $V(\Ly-Z)=V(Z)$ (with $Z=Y-s$), one obtains $V(\Ly-Y+s)=V(Y-s)$. Substituting back,
\begin{equation}
    h(-s)\,\tilde\phi_\nu = \epnu(s)\,\tilde\phi_\nu.
\end{equation}
This shows that $\tilde\phi_\nu$ is an eigenstate of $h(-s)$ with eigenvalue $\epnu(s)$. However, this conclusion only states that $\epnu(s)$ belongs to the spectrum of $h(-s)$, i.e., to the set of eigenvalues $\mathrm{spec}[h(-s)]\equiv\{\epsilon_1(-s),\epsilon_2(-s),\ldots\}$; it does not imply that $\epnu(s)$ equals the \emph{$\nu$th} eigenvalue of $h(-s)$. In general, $\epnu(s)=\epsilon_\mu(-s)$ for some label $\mu$, which need not equal $\nu$: although $\tilde\phi_\nu$ is constructed from $\phi_\nu$, it may correspond to any band of $h(-s)$. Establishing $\mu=\nu$ is therefore the crucial step of the proof.

At $s=0$, $h(s)=h(-s)=h(0)$, and the symmetry condition $V(\Ly-Z)=V(Z)$ endows the potential of $h(0)$ with reflection symmetry about $Y_c=\Ly/2$. In this case $\tilde\phi_\nu(Y)=\phi_\nu(\Ly-Y;0)$ is also an eigenstate of $h(0)$ [with eigenvalue $\epnu(0)$]. The nondegeneracy condition guarantees that each eigenvalue of $h(0)$ corresponds to a single eigenstate (up to a proportionality constant), so $\tilde\phi_\nu\propto\phi_\nu$, i.e., $\mu=\nu$ holds at $s=0$.

For $s\neq 0$, this label matching must be extended to the entire real axis. By Kato's analytic perturbation theory~\cite{Kato1966}, under the nondegeneracy condition $\epnu(s)$ is an analytic function of $s$, and each $\nu$ labels a smooth eigenvalue curve. An exchange of level labels can occur only at an exceptional point (EP), i.e., where two eigenvalue curves merge. The theorem assumes that $h(s)$ is nondegenerate along the entire real axis $s\in\mathbb{R}$, excluding the occurrence of EPs. Starting from $s=0$ (where $\mu=\nu$ has been established) and moving continuously along the real axis to arbitrary $s$, the level labels remain $\mu=\nu$ throughout, and finally $\epnu(-s)=\epnu(s)$ holds for all $\nu,s$.

\textbf{(iii) Violation of the symmetry condition $\Rightarrow$ NCSE for generic potentials.}

If $V(\Ly-Z)\neq V(Z)$ for some $Z$, the similarity transformation in (ii) is broken and $\epnu(s)=\epnu(-s)$ loses its protection. More quantitatively, in the weak-potential limit, the simultaneous vanishing of $\mathrm{Re}[\alpha_\nu]=\langle\nu|U'|\nu\rangle$ for all $\nu$ would mean that $U'(Y)$ is orthogonal to all $\sin^2(\nu\pi Y/\Ly)=(1-\cos(2\nu\pi Y/\Ly))/2$: this requires the mean value $a_0$ and all even coefficients $a_{2\nu}$ ($\nu=1,2,\ldots$) of the Fourier-cosine expansion of $U'(Y)$ to vanish, which is precisely equivalent to $U'(Y)$ being antisymmetric about $\Ly/2$, i.e., to the reflection-symmetry condition $U(\Ly-Y)=U(Y)$. Therefore, in the weak-potential limit, $\alpha_\nu=0$ for all $\nu$ holds \emph{if and only if} the potential satisfies the reflection-symmetry condition, and any potential violating the symmetry condition necessarily has some $\nu$ with $\mathrm{Re}[\alpha_\nu]\neq 0$. At general potential strength, $[\phi_\nu^R]^2$ deviates from $|\phi_\nu^{(0)}|^2$ and the exact equivalence above no longer holds strictly, but an analogous Fourier analysis shows that for generic potentials there still exists some $\nu$ with $\mathrm{Re}[\alpha_\nu]\neq 0$.

\subsection{Validity of the adiabatic-Taylor reduction}\label{app:validity}
Both the Taylor truncation and the adiabatic approximation presuppose the nondegeneracy of $h(s)$. Near an exceptional point (EP), two eigenvalues and their eigenstates coalesce simultaneously, the sensitivity of $\phi_\nu$ to $s$ diverges, and the adiabatic condition fails locally; at the same time, the breakdown of nondegeneracy also affects the $\mu=\nu$ argument in Theorem~\ref{thm:symmetry}(ii). The behavior of the NCSE near an EP may deviate from the predictions of the Hatano--Nelson mapping.

Specifically, Theorem~\ref{thm:symmetry}(ii) requires $h(s)$ to be free of EPs along the entire real axis $s\in\mathbb{R}$. If two levels merge at some $s=s_{\mathrm{EP}}$, analytic continuation can no longer exclude a label exchange [$\epnu(s)$ may jump to another branch at the EP], and $\epnu(s)=\epnu(-s)$ may fail. Whether the NCSE remains absent when the symmetry condition is satisfied but EPs are present is an open question.

\section{Effective One-Dimensional Equation and Similarity Transformation}\label{app:exact-1D}
This appendix provides the precise definition of the pseudo-differential operator $\epnu(-i\theta\partial_X)$ and an independent verification of the similarity-transformation solution of Eq.~\eqref{eq:HN}.

The operator $-i\partial_X$ acts on a plane wave $e^{i\beta X}$ as $(-i\partial_X)e^{i\beta X}=\beta\,e^{i\beta X}$. For an analytic function $\epnu(s)$, the pseudo-differential operator $\epnu(-i\theta\partial_X)$ is defined through its spectral action~\cite{Zworski2012}:
\begin{equation}\label{eq:spectral-def}
\epnu(-i\theta\partial_X)\,e^{i\beta X}\equiv\epnu(\theta\beta)\,e^{i\beta X}.
\end{equation}
For a general function $f(X)=\int\hat{f}(\beta)e^{i\beta X}d\beta$, one has $\epnu(-i\theta\partial_X)f=\int\hat{f}(\beta)\epnu(\theta\beta)e^{i\beta X}d\beta$, i.e., a Fourier multiplier, which in position space is generally a nonlocal (convolution) operator. When $\epnu(s)$ is a polynomial, it reduces to a standard differential operator. The Taylor truncation $\epnu(s)\approx\epnu(0)+i\alpha_\nu s$ converts the Fourier multiplier into a first-order differential operator, $\epnu(-i\theta\partial_X)\approx\epnu(0)+\alpha_\nu\theta\partial_X$, which is local. Substituting into Eq.~\eqref{eq:exact-1D} yields the Hatano--Nelson equation, Eq.~\eqref{eq:HN}.

The eigenstates Eq.~\eqref{eq:chi-OBC} and eigenvalues Eq.~\eqref{eq:E-OBC} of Eq.~\eqref{eq:HN} can also be verified independently through a similarity transformation. Define $S\equiv e^{-\alpha_\nu\theta X/2}$ (a multiplication operator), with $S^{-1}=e^{\alpha_\nu\theta X/2}$. Since $S$ is nonzero at $X=0,\Lx$, it preserves the OBC. Let $\tilde\psi\equiv S\psi$. Using $\partial_X[e^{\alpha_\nu\theta X/2}f]=(\alpha_\nu\theta/2)e^{\alpha_\nu\theta X/2}f+e^{\alpha_\nu\theta X/2}f'$, direct expansion gives
\begin{equation}\label{eq:S-kinetic}
S(-\partial_X^2)S^{-1}=-\partial_X^2-\alpha_\nu\theta\partial_X-\tfrac{\alpha_\nu^2\theta^2}{4}\,,
\end{equation}
together with $S(\alpha_\nu\theta\partial_X)S^{-1}=\alpha_\nu\theta\partial_X+\alpha_\nu^2\theta^2/2$. Combining the two, the linear $\partial_X$ terms cancel exactly, and the transformed equation reads
\begin{equation}\label{eq:H-tilde}
-\tilde\psi''+\!\left(\epnu(0)+\tfrac{\alpha_\nu^2\theta^2}{4}\right)\tilde\psi=E\tilde\psi.
\end{equation}
This is a free particle in a box under OBC, with eigenstates $\tilde\psi_n\propto\sin(n\pi X/\Lx)$; the inverse transformation $\psi_n=S^{-1}\tilde\psi_n=e^{\alpha_\nu\theta X/2}\tilde\psi_n$ recovers Eq.~\eqref{eq:chi-OBC}. Physically, $S$ absorbs the imaginary and real gauge fields simultaneously, mapping the non-Hermitian Hatano--Nelson problem exactly onto a Hermitian particle in a box.

\section{Lattice Hamiltonian and Potential-Specific Derivations}\label{app:lattice-general}

We take a square lattice $(X_m,Y_n)=(ma,na)$ with $m=1,\ldots,N_x$, $n=1,\ldots,N_y$, and lattice constant $a=1$. Under OBC the boundaries lie at $n=0$ and $n=N_y+1$, and the effective box length along $Y$ is $\Ly=(N_y+1)a$. The discretization of the continuum Hamiltonian Eq.~\eqref{eq:Bopp-H} proceeds in two steps: the kinetic energy is discretized by standard central differences, and the potential is discretized order by order through the Bopp expansion Eq.~\eqref{eq:Bopp-expansion}. The standard discretization of the kinetic term $\hat{P}_X^2+\hat{P}_Y^2$ gives the tight-binding Hamiltonian (dropping the constant on-site energy $\varepsilon_0=4t$)
\begin{equation}\label{eq:kinetic-lattice}
H_0=-t\!\sum_{m,n}\bigl(c_{m,n}^\dagger c_{m+1,n}+c_{m,n}^\dagger c_{m,n+1}+\mathrm{h.c.}\bigr),
\end{equation}
where $t\equiv 1/a^2$ is the hopping integral and $c_{m,n}^\dagger$ ($c_{m,n}$) is the creation (annihilation) operator on site $(m,n)$. For a purely imaginary potential $V=iU$ ($U$ real), the $k$th-order term of the Bopp expansion Eq.~\eqref{eq:Bopp-expansion} reads, in the coordinate representation,
\begin{equation}\label{eq:Bopp-k-general}
\frac{1}{k!}V^{(k)}(Y)(\theta\hat{P}_X)^k=\frac{(-1)^k\,i^{k+1}}{k!}\,U^{(k)}(Y)\,\theta^k\,\partial_X^k\,,
\end{equation}
where $V^{(k)}=iU^{(k)}$ supplies a factor of $i$ and $\hat{P}_X^k=(-i\partial_X)^k$ supplies $(-i)^k$; their product $(-1)^k i^{k+1}$ rigorously determines the physical character of each order. For odd $k$, $i^{k+1}$ is real, and a real coefficient multiplying the antisymmetric structure of an odd-order difference produces a real difference between the left and right hopping amplitudes (non-reciprocal hopping), the source of the NCSE. For even $k$, $i^{k+1}$ is purely imaginary, and the imaginary coefficient acting on an even-order (symmetric) difference endows the left and right hoppings with identical imaginary corrections, leaving reciprocity intact. This even-odd rule is a direct algebraic consequence of the purely imaginary potential $V=iU$ combined with the structure of the Bopp expansion.

The lattice discretization of each order is given by the standard central-difference formulas:
\begin{equation}\label{eq:d1}
\begin{aligned}
\partial_X\psi_m&=\frac{\psi_{m+1}-\psi_{m-1}}{2a},\\
\partial_X^2\psi_m&=\frac{\psi_{m+1}-2\psi_m+\psi_{m-1}}{a^2},\\
\partial_X^3\psi_m&=\frac{\psi_{m+2}-2\psi_{m+1}+2\psi_{m-1}-\psi_{m-2}}{2a^3}.
\end{aligned}
\end{equation}
Odd-order differences are antisymmetric (producing non-reciprocal hopping) and even-order ones are symmetric (preserving reciprocity). Defining the correction amplitudes of each order,
\begin{equation}\label{eq:k-lattice}
\delta t_n^{(1)}\equiv\frac{\theta U'}{2a},\;\delta t_n^{(2)}\equiv\frac{-i\theta^2 U''}{2a^2},\;\delta t_n^{(3)}\equiv\frac{\theta^3 U'''}{12a^3},
\end{equation}
the $k=1$ order produces nearest-neighbor non-reciprocal hopping $\pm\delta t_n^{(1)}$; the $k=2$ order produces the symmetric correction $\delta t_n^{(2)}$ and the on-site correction $-2\delta t_n^{(2)}$; the $k=3$ order contributes simultaneously to nearest-neighbor hopping, $\pm 2\delta t_n^{(3)}$, and next-nearest-neighbor hopping, $\mp\delta t_n^{(3)}$. Combining the kinetic term Eq.~\eqref{eq:kinetic-lattice} with the first three Bopp orders, the generic lattice Hamiltonian reads (dropping $\varepsilon_0$)
\begin{equation}\label{eq:lattice-general}
\begin{aligned}
H=H_0&+\!\sum_{m,n}\bigl[\delta J_n^+\,c_{m,n}^\dagger c_{m+1,n}+\delta J_n^-\,c_{m,n}^\dagger c_{m-1,n}\\
&+\left(V(Y_n)-2\delta t_n^{(2)}\right)\,c_{m,n}^\dagger c_{m,n}\bigr],
\end{aligned}
\end{equation}
where the Bopp hopping corrections are $\delta J_n^\pm\equiv\pm(\delta t_n^{(1)}+2\delta t_n^{(3)})+\delta t_n^{(2)}$, and the total effective hopping amplitudes along $X$ are $J_n^\pm\equiv -t+\delta J_n^\pm$. The non-reciprocal part, $J_n^+-J_n^-=2(\delta t_n^{(1)}+2\delta t_n^{(3)})$, comes entirely from the odd Bopp orders; the symmetric part, $(J_n^++J_n^-)/2+t=\delta t_n^{(2)}$, comes from the even orders. At small $\theta$, the $k=1$ nearest-neighbor non-reciprocal hopping $\delta t_n^{(1)}\propto\theta$ dominates the NCSE: for the cubic potential, the ratio of the additional $k=3$ nearest-neighbor contribution $2\delta t_n^{(3)}$ to $\delta t_n^{(1)}$ is $O(\theta^2/N_y^2)$ (see Appendix~\ref{app:NNN-estimate} for the derivation), negligible for $\theta\ll N_y$. For a polynomial potential of degree $N$, $U^{(k)}\equiv 0$ for $k>N$, and the lattice Hamiltonian carries no truncation error (linear potential, $k\leq 1$; quadratic, $k\leq 2$; cubic, $k\leq 3$). This is a level of expansion independent of the Taylor truncation of $\epnu(s)$ in Theorem~\ref{thm:NHSE}: for polynomial potentials the former terminates exactly (even for $\theta\gg 1$) and involves no approximation, while the latter becomes quantitatively exact in the thermodynamic limit $\Lx\to\infty$. In the continuum limit $a\to 0$, the lattice model recovers the continuum theory of Sec.~\ref{sec:quantitative}.

The following subsections present the complete derivations for each specific imaginary potential, tracing the chain from the Bopp expansion to the inverse skin length. The Bopp expansion orders and lattice correction terms for the polynomial potentials are summarized in Table~\ref{tab:polynomial-summary}.
\begin{table}[h]
\centering
\caption{Bopp expansion for polynomial potentials ($\eta_n\equiv Y_n-\Ly/2$).}
\label{tab:polynomial-summary}
\begin{ruledtabular}
\begin{tabular}{lcccc}
Potential & Order & $\delta t_n^{(1)}$ & $\delta t_n^{(2)}$ & $\delta t_n^{(3)}$ \\
\hline
$igY$ & 1 & $g\theta/(2a)$ & 0 & 0 \\
$ig\eta_n^2$ & 2 & $g\theta\eta_n/a$ & $-ig\theta^2/a^2$ & 0 \\
$ig\eta_n^3$ & 3 & $3g\theta\eta_n^2/(2a)$ & $-3ig\theta^2\eta_n/a^2$ & $g\theta^3/(2a^3)$ \\
\end{tabular}
\end{ruledtabular}
\end{table}

\subsection{Quadratic potential: $V=ig(Y-Y_c)^2$}\label{app:quadratic}

Define $\hat{\eta}\equiv\hat{Y}-Y_c$ with $Y_c=\Ly/2$. The potential and its derivatives are $V'=2ig\hat{\eta}$, $V''=2ig$, and $V^{(k)}=0$ for $k\geq 3$. The Bopp expansion terminates exactly at $k=2$:
\begin{equation}\label{eq:quad-bopp-full}
V(\hat{Y}+\theta\hat{P}_X)=ig\hat{\eta}^2+2ig\hat{\eta}\,\theta\hat{P}_X+ig\theta^2\hat{P}_X^2\,,
\end{equation}
i.e., the three terms of the binomial expansion of $ig(\hat{\eta}+\theta\hat{P}_X)^2$. In the coordinate representation, the three terms are: $k=0$, the on-site imaginary potential $ig(Y-Y_c)^2$; $k=1$, the non-reciprocal hopping $2g(Y-Y_c)\theta\partial_X$; $k=2$, the symmetric kinetic correction $-ig\theta^2\partial_X^2$. The full Hamiltonian is
\begin{equation}\label{eq:H-quadratic}
H_{\mathrm{quad}}=-(\partial_X^2+\partial_Y^2)+ig(\hat{\eta}-i\theta\partial_X)^2\,,
\end{equation}
whose expansion is the coordinate-representation form of Eq.~\eqref{eq:quad-bopp-full}. The $k=1$ term $2g\hat{\eta}\theta\partial_X$ contains both $Y$ and $\partial_X$, so the $X$ and $Y$ directions do not decouple, in contrast to the complete separation for the linear potential. At the same time, since $\alpha_\nu=\langle\nu|U'|\nu\rangle=2g\langle\nu|(Y-\Ly/2)|\nu\rangle=0$ (the integrand is an odd function about $\Ly/2$), the weak-potential limit gives $\kappa_\nu=0$.

\subsection{Cubic potential: $V=ig(Y-Y_c)^3$}\label{app:cubic}

The potential and its derivatives are $V'=3ig\hat{\eta}^2$, $V''=6ig\hat{\eta}$, $V'''=6ig$, and $V^{(k)}=0$ for $k\geq 4$. The Bopp expansion terminates exactly at $k=3$:
\begin{equation}\label{eq:cub-bopp-full}
V(\hat{Y}+\theta\hat{P}_X)=ig\hat{\eta}^3+3ig\hat{\eta}^2\theta\hat{P}_X+3ig\hat{\eta}\theta^2\hat{P}_X^2+ig\theta^3\hat{P}_X^3,
\end{equation}
i.e., $ig(\hat{\eta}+\theta\hat{P}_X)^3$. The full Hamiltonian (coordinate representation) is
\begin{equation}\label{eq:H-cubic}
H_{\mathrm{cub}}=-(\partial_X^2+\partial_Y^2)+ig(\hat{\eta}-i\theta\partial_X)^3\,,
\end{equation}
whose expansion is the coordinate-representation form of Eq.~\eqref{eq:cub-bopp-full}. The effects of the various orders in the coordinate representation are: $k=1$, $3g(Y-Y_c)^2\theta\partial_X$, non-reciprocal hopping with strength proportional to $(Y-Y_c)^2$; $k=2$, $-3ig(Y-Y_c)\theta^2\partial_X^2$, a $Y$-dependent symmetric correction; $k=3$, $-g\theta^3\partial_X^3$, uniform next-nearest-neighbor non-reciprocal hopping.

For the non-reciprocity coefficient, $\alpha_\nu=3g\langle\nu|(Y-\Ly/2)^2|\nu\rangle\equiv 3g\sigma_\nu^2$. Setting $u=Y/\Ly$, one has $\sigma_\nu^2=\Ly^2\langle(u-\tfrac{1}{2})^2\rangle$, with the expectation value weighted by $|\phi_\nu^{(0)}|^2=\frac{2}{\Ly}\sin^2(\nu\pi u)$. Expanding $\langle(u-\tfrac{1}{2})^2\rangle=\langle u^2\rangle-\langle u\rangle+\tfrac{1}{4}$ and simplifying the integrals using $\sin^2=(1-\cos)/2$, a direct calculation gives
\begin{equation}
\sigma_\nu^2=\frac{\Ly^2}{12}\!\left(1-\frac{6}{\nu^2\pi^2}\right),
\end{equation}
recovering Eq.~\eqref{eq:sigma-nu} of the main text. The inverse skin length is $\kappa_\nu=\frac{3g\theta}{2}\sigma_\nu^2=\frac{g\theta\Ly^2}{8}(1-6/(\nu^2\pi^2))$.

\subsection{Dominance of nearest-neighbor non-reciprocity at small $\theta$}\label{app:NNN-estimate}

For polynomial potentials of degree $N\geq 3$, the $k=3$ Bopp correction contributes simultaneously to nearest-neighbor and next-nearest-neighbor hopping. This section estimates its magnitude relative to the leading $k=1$ term, providing quantitative justification for neglecting the higher-order non-reciprocal terms at small $\theta$ in Sec.~\ref{sec:numerical-verification} of the main text.

Take the cubic potential ($a=1$, $t=1$) as an example. The additional $k=3$ contribution to nearest-neighbor hopping is $2\delta t_n^{(3)}=g\theta^3$, while the $k=1$ term is $\delta t_n^{(1)}=3g\theta\eta_n^2/2$. Their ratio at a typical site located a quarter of the interval away from the center, $\eta_n\sim N_y/4$, is
\begin{equation}\label{eq:NNN-ratio}
\frac{2\delta t_n^{(3)}}{\delta t_n^{(1)}}\bigg|_{\eta_n\sim N_y/4}\sim\frac{32\theta^2}{3N_y^2}\,.
\end{equation}
This ratio is far smaller than unity for $\theta\ll N_y$; in addition, the hopping corrections themselves must remain small compared with the kinetic hopping: $\delta t_n^{(3)}/t \ll 1$.

The key physical point is that for polynomial potentials the inverse skin length grows linearly, $\kappa\propto g\theta$, without the geometric suppression $|\kappa|\sim g/\omega^3$ of the sine potential (Sec.~\ref{sec:peierls}), so $\theta\sim O(1)$ already produces a pronounced skin effect with $\kappa N_x\gg 1$. For example, for $g=10^{-3}$, $N_y=50$, $\theta=1$: $\kappa_1 N_x\approx 6$ (clearly visible), while $\delta t_n^{(3)}/t\approx 5\times 10^{-4}$ (entirely negligible). Focusing on the nearest-neighbor non-reciprocal hopping in the small-$\theta$ numerical verification is therefore self-consistent.

\subsection{Matrix element for sine potential}\label{app:sine-matrix}

According to Eq.~\eqref{eq:kappa-sine-derivation}, the inverse skin length of the sine potential is determined by the matrix element $I_\nu(\omega)=\frac{2}{\Ly}\int_0^{\Ly}\sin^2(\nu\pi Y/\Ly)\cos(\omega Y/\Ly)\,dY$. Using $\sin^2 x=(1-\cos 2x)/2$, the integral splits into two parts. The first part gives $\sin\omega/\omega$; the second, using product-to-sum identities together with $\sin(2\nu\pi\pm\omega)=\pm\sin\omega$, simplifies to $-\omega\sin\omega/(4\nu^2\pi^2-\omega^2)$ [this is the value of the $\cos$--$\cos$ integral; the minus sign in $\sin^2=(1-\cos)/2$ makes its contribution to $I_\nu$ positive]. Combining,
\begin{equation}\label{eq:Inu-result}
I_\nu(\omega)=\frac{4\nu^2\pi^2\sin\omega}{\omega(4\nu^2\pi^2-\omega^2)},
\end{equation}
valid for all $\nu\geq 1$ and $\omega\neq 2\nu\pi$. The denominator $4\nu^2\pi^2-\omega^2$ vanishes at $\omega=2\nu\pi$, but this singularity is removable: by L'H\^opital's rule, or returning directly to the integral, the factor $\cos((2\nu\pi-\omega)Y/\Ly)$ in the product of $\sin^2(\nu\pi Y/\Ly)$ and $\cos(2\nu\pi Y/\Ly)$ degenerates into a constant as $\omega\to 2\nu\pi$, and the integral yields the finite value $I_\nu(2\nu\pi)=-1/2$. As a function of $\omega$, $\kappa_\nu$ is everywhere finite and continuous, with no physical singularity. For $\omega<2\pi$, the sign of $\kappa_\nu$ for all modes is determined by $\sin\omega$.

\subsection{Cosine potential}\label{app:cosine}

The Peierls mapping of the cosine imaginary potential $V(\hat{y})=ig\cos(\hat{y}/\ell)$ exactly parallels that of the sine potential. Using $ig\cos(\hat{y}/\ell)=\frac{ig}{2}(e^{i\hat{y}/\ell}+e^{-i\hat{y}/\ell})$ (the two terms carry the same sign), under the commensurability condition:
\begin{equation}\label{eq:cosine-lattice}
V(\hat{y})\psi_{m,n}=\frac{ig}{2}\Big(e^{in\phi}\psi_{m+1,n}+e^{-in\phi}\psi_{m-1,n}\Big).
\end{equation}
Although the individual hopping corrections $\frac{ig}{2}e^{\pm in\phi}$ contain imaginary factors, the non-reciprocal part driving the NCSE (the difference between left and right hoppings) remains real: $J_+(n)-J_-(n)=\frac{ig}{2}(e^{in\phi}-e^{-in\phi})=-g\sin(n\phi)\in\mathbb{R}$, a real non-reciprocal hopping just like the $g\cos(n\phi)$ of the sine potential. The difference between the two lies not in the real/imaginary character of the non-reciprocity, but in the different trigonometric dependence on the lattice index $n$.

Applying the perturbative criterion [$V'=-ig\sin(Y/\ell)/\ell$], the matrix element is
\begin{equation}\label{eq:Jnu-result}
J_\nu(\omega)=\frac{4\nu^2\pi^2(1-\cos\omega)}{\omega(4\nu^2\pi^2-\omega^2)},
\end{equation}
and the inverse skin length is
\begin{equation}\label{eq:kappa-cosine}
\kappa_\nu^{(\cos)}=-\frac{2ga_0\nu^2\pi^2(1-\cos\omega)}{\omega(4\nu^2\pi^2-\omega^2)}.
\end{equation}
The symmetry condition $V(\Ly-Z)=V(Z)$ is equivalent to $\omega=2k\pi$, at which $1-\cos(2k\pi)=0$ and $\kappa_\nu^{(\cos)}=0$, consistent with Theorem~\ref{thm:symmetry}.

The core difference between the two cases lies in the trigonometric factors in the numerator, $\sin\omega$ (sine potential) versus $(1-\cos\omega)$ (cosine potential), and manifests itself mainly in the behavior near the symmetry points. For the sine potential, setting $\delta=\omega-\omega_c$ near a symmetry point $\omega_c=(2k+1)\pi$, one has $\sin\omega\approx-\delta$: $\kappa_\nu$ crosses zero linearly, and a sign change of $\delta$ flips the skin direction of all modes, precisely the collective direction-flip effect of Sec.~\ref{sec:peierls}. For the cosine potential, near a symmetry point $\omega_c=2k\pi$ the numerator factor $1-\cos\omega\approx\delta^2/2\geq 0$ is non-negative: it touches zero quadratically without changing sign. The denominator $4\nu^2\pi^2-\omega^2$ changes sign at $\omega=2\nu\pi$, allowing the $\kappa_\nu$ of the $\nu$th band to flip individually at $\omega=2\nu\pi$. This flip, however, is band dependent (different bands flip at different $\omega$), in contrast to the sine potential, where all bands flip collectively at the same $\omega_c$. For the total density $\bar\rho(m)=\sum_\nu\bar\rho_\nu(m)$, when $\omega$ crosses $2\pi$ only the first band undergoes a weak flip ($|\kappa_1|\sim g|\delta|/8$) while the remaining $N_y-1$ bands keep their original direction, so the direction of the total density is unchanged: no collective direction flip exists.

The denominators $4\nu^2\pi^2-\omega^2$ of Eqs.~\eqref{eq:Inu-result} and \eqref{eq:Jnu-result} vanish at $\omega=2\nu\pi$. As in the case of the sine matrix element, these singularities are removable: for the sine potential, $I_\nu(2\nu\pi)=-1/2$; for the cosine potential, $J_\nu(2\nu\pi)=0$ [since the numerator $1-\cos(2\nu\pi)=0$ makes numerator and denominator vanish simultaneously, with a vanishing limit]. Both formulas are everywhere finite and continuous as functions of $\omega$.

\subsection{Parameter constraints for numerical verification}\label{app:param-constraints}

The band-resolved numerical verification of $\kappa_\nu$ in Theorem~\ref{thm:NHSE} relies on several conditions holding simultaneously: the weak-potential condition (Appendix~\ref{app:weak-potential}) ensures the validity of the closed-form expression Eq.~\eqref{eq:kappa-perturbative}; the adiabatic condition (Appendix~\ref{app:adiabatic}) ensures the approximate separability of the two-dimensional eigenstates, and its global version ($j_{\mathrm{safe}}\gtrsim N_x$), combined with the visibility condition, yields for the cubic potential the scaling law $N_x\gtrsim 10(N_y-1)^2$. In addition, the following two lattice-level conditions are required.

\textbf{Skin visibility}: the exponential envelope must be resolvable in a finite system, $|\kappa_\nu|N_x\gtrsim 1$. This condition sets lower bounds on the system parameters. For the linear potential, the $X$ and $Y$ directions decouple completely and $\kappa=g\theta/2$ holds exactly for arbitrary parameters, so only this condition is needed. For the quadratic potential, $\kappa_\nu\equiv 0$ (guaranteed by Theorem~\ref{thm:symmetry}), and no $\kappa$ extraction is required.

\textbf{Bopp-expansion convergence} (polynomial potentials only): the higher-order Bopp corrections in the lattice Hamiltonian must remain small compared with the kinetic hopping. For the cubic potential, both the nearest-neighbor condition $3g\theta\eta_{\max}^2/2\ll t$ and the next-nearest-neighbor condition $g\theta^3/(2a^3)\ll t$ are required; the latter is often the most stringent constraint at large $\theta$. The Peierls mappings of the sine and cosine potentials are exact, and this condition is satisfied automatically.

Table~\ref{tab:param-constraints} summarizes the visibility and Bopp-convergence constraints for each potential. The essential difference is that for polynomial potentials $\max|U|$ grows as a power of $\Ly$, so the parameter window shrinks rapidly at large $N_y$ (for the cubic potential, $N_y$ should not exceed ${\sim}15$), whereas for the sine and cosine potentials $|U|\leq g$ is globally bounded and the Peierls mapping is exact, leaving the widest freedom in parameter choice.

\begin{table}[h]
\centering
\caption{Lattice-level conditions for the numerical verification of each imaginary potential. ``---'' denotes automatically satisfied or not applicable. The weak-potential and adiabatic conditions are given in Appendixes~\ref{app:weak-potential} and \ref{app:adiabatic}. $\eta_{\max}=(N_y-1)/2$, $a=1$.}
\label{tab:param-constraints}
\begin{ruledtabular}
\begin{tabular}{lcc}
Imaginary potential & Visibility & Bopp convergence \\
\hline
$igY$ & $g\theta N_x/2\gtrsim 1$ & --- \\
$ig(Y{-}Y_c)^2$ & --- & --- \\
$ig(Y{-}Y_c)^3$ & $g\theta\Ly^2 N_x/8\gtrsim 1$ & $g\theta^3\ll 2$ \\
$ig\sin(Y/\ell)$ & Eq.~\eqref{eq:kappa-sine} & --- \\
$ig\cos(Y/\ell)$ & Eq.~\eqref{eq:kappa-cosine} & --- \\
\end{tabular}
\end{ruledtabular}
\end{table}

\section{Physical Realization: Detailed Analysis}\label{app:realization}

The mechanism of this work requires three ingredients to coexist: (R1) coordinate noncommutativity $[\hat{x},\hat{y}]=i\theta$; (R2) standard second-order kinetic energy $\hat{P}_X^2+\hat{P}_Y^2$; and (R3) an imaginary potential $V(\hat{y})=iU(\hat{y})$. Below we analyze in detail the structural obstructions faced by the two most direct candidate platforms (synthetic-dimension schemes are discussed in Sec.~\ref{sec:discussion} of the main text).

\subsection{Lowest-Landau-level guiding center}\label{app:LLL}

In a strong magnetic field $\bm{B}=B\hat{z}$ (symmetric gauge), the guiding-center coordinates $\hat{X}_c=\hat{x}-\hat{\pi}_y/(m\omega_c)$, $\hat{Y}_c=\hat{y}+\hat{\pi}_x/(m\omega_c)$ [$\hat{\pi}_i=\hat{p}_i+eA_i$ are the dynamical momenta and $\omega_c=eB/m$ is the cyclotron frequency] satisfy $[\hat{X}_c,\hat{Y}_c]=i\ell_B^2$~\cite{Ezawa2013}. The Landau levels are $E_n=\hbar\omega_c(n+1/2)$. Defining the LLL projector $\mathcal{P}_0\equiv|n=0\rangle\langle n=0|$ (tracing out the cyclotron degree of freedom), and using the Landau Hamiltonian $H_0\equiv\hbar\omega_c(\hat{a}^\dagger\hat{a}+1/2)$ ($\hat{a}$, $\hat{a}^\dagger$ are the ladder operators of the cyclotron motion) together with $\hat{a}|0\rangle=0$, one finds $\mathcal{P}_0 H_0\mathcal{P}_0=(\hbar\omega_c/2)\mathcal{P}_0$, so the effective Hamiltonian is
\begin{equation}
H_{\mathrm{eff}}=\frac{\hbar\omega_c}{2}\mathcal{P}_0+\mathcal{P}_0 V(\hat{y})\mathcal{P}_0,
\end{equation}
in which the kinetic energy degenerates into a constant independent of the guiding-center position; there is no second-order dispersion. The NCSE structure of this work, $H_X=\hat{P}_X^2+ig\theta\hat{P}_X$, relies on the competition between the second-order kinetic energy and the first-order imaginary gauge field; in the LLL this competition is absent, and ingredient (R2) is missing.

Even without the LLL projection, i.e., retaining all Landau levels, Landau quantization still profoundly modifies the kinetic structure, and recent studies have shown that magnetic fields generically suppress the NHSE: Lu, Zhang, and Franz~\cite{Lu2021} proved rigorously, in a two-dimensional Hatano--Nelson lattice model, that a magnetic field pushes the skin states from the boundary back into the bulk, drastically shrinking the ``skin-topological area,'' the physical origin being precisely the competition between the bulk-localized states produced by Landau quantization and the skin boundary states. Subsequent work extended this effect to pseudomagnetic fields~\cite{Teo2024} and to more general Harper--Hofstadter models~\cite{Longhi2025}. Notably, Li \emph{et al.}~\cite{Li2023magnetic} found that while a magnetic field suppresses the first-order NHSE, it can enhance the second-order skin effect, indicating that the interplay between magnetic fields and non-Hermitian skin physics is richer than simple ``suppression.''

\subsection{Phase-space reinterpretation}

Standard one-dimensional quantum mechanics, $[\hat{x},\hat{p}]=i\hbar$, is algebraically isomorphic to $[\hat{x},\hat{y}]=i\theta$. Consider $H=p^2/(2m)+igx$. In the momentum representation, $\hat{x}=i\hbar\partial_p$ and the Hamiltonian becomes $H=p^2/(2m)-g\hbar\partial_p$. Superficially, the first-order term $-g\hbar\partial_p$ resembles an imaginary gauge field, but no similarity transformation can remove it. Let $\chi(p)=e^{\lambda p}\phi(p)$ and substitute into $H\chi=E\chi$; since $p^2/(2m)$ is a multiplication operator in the momentum representation, $e^{\lambda p}$ cancels through it directly without generating any cross derivative terms:
\begin{equation}
\frac{p^2}{2m}\phi-g\hbar(\lambda\phi+\phi')=E\phi.
\end{equation}
The first-derivative term $-g\hbar\phi'$ cannot be eliminated, while $\lambda$ contributes only an energy shift $-g\hbar\lambda$, for any choice of $\lambda$. Contrast this with $H_X=-\partial_X^2+g\theta\partial_X$ of this work: both terms are \emph{differential} operators, and the exponential transformation $\psi=e^{\kappa X}\phi$ renders the total coefficient of $\partial_X$ equal to $-2\kappa+g\theta$, which can be set to zero ($\kappa=g\theta/2$), eliminating the first-order term exactly. The root cause is a mismatch of phase-space dimensions. The phase space of standard quantum mechanics is the two-dimensional $(\hat{x},\hat{p})$: the imaginary potential $igx$ becomes, after Fourier transformation, a differential operator $-g\hbar\partial_p$ in momentum space, but it cannot mathematically mix with the kinetic term $p^2/(2m)$, which is a multiplication operator: one is multiplicative, the other differential. The noncommutative space of this work has a four-dimensional phase space $(\hat{x},\hat{y},\hat{p}_x,\hat{p}_y)$, and the Bopp shift $\hat{y}=\hat{Y}+\theta\hat{P}_X$ is not a Fourier transformation but a representation of $\hat{y}$ as a linear combination of a coordinate and a momentum~\cite{Douglas2001,Szabo2003}. The imaginary-potential term $ig\theta\hat{P}_X$ and the kinetic term $\hat{P}_X^2$ act on the same variable $X$ and are both differential operators; a similarity transformation can therefore connect the two, which is the structural prerequisite for the emergence of the NHSE.


%

\end{document}